\documentclass[letterpaper,11pt,twoside]{article}
\usepackage{prestuff}

\newcommand{\MeetInTheMiddle}{\text{\tt Meet-in-the-Middle}}
\newcommand{\SSE}{\text{\tt Sorted-Sum-Enumeration}}
\newcommand{\BitPacking}{\text{\tt Bit-Packing}}
\newcommand{\RepresentationOV}{\text{\tt Representation-OV}}
\newcommand{\PackedRepresentationOV}{\text{\tt Packed-Representation-OV}}

\newcommand{\word}{\ell}
\newcommand{\algcomment}[1]{\hfill \OliveGreen{$\triangleright$ #1}}

\title{Subset Sum in Time $2^{n/2} / \poly(n)$}

\author{Xi Chen, Yaonan Jin, Tim Randolph, and Rocco A.~Servedio 
\\ \texttt{ \{xichen, rocco\}@cs.columbia.edu, \{yj2552, t.randolph\}@columbia.edu }
\\ Columbia University}
\date{\today}

\begin{document}

\maketitle

\begin{abstract}
A major goal in the area of exact exponential algorithms is to give an algorithm for the (worst-case) $n$-input Subset Sum problem that runs in time $2^{(1/2 - c)n}$ for some constant $c>0$.
In this paper we give a Subset Sum algorithm with worst-case running time $O(2^{n/2} \cdot n^{-\gamma})$ for a constant $\gamma > 0.5023$ in standard word RAM or circuit RAM models.
To the best of our knowledge, this is the first improvement on the classical ``meet-in-the-middle'' algorithm for worst-case Subset Sum, due to Horowitz and Sahni, which can be implemented in time $O(2^{n/2})$ in these memory models \cite{horowitz1974computing}.

Our algorithm combines a number of different techniques, including the ``representation method'' introduced by  Howgrave-Graham and Joux \cite{howgrave2010new} and subsequent adaptations of the method in Austrin, Kaski, Koivisto, and Nederlof~\cite{austrin2016dense}, and Nederlof and Wegrzycki \cite{NederlofW21}, and ``bit-packing'' techniques used in the work of Baran, Demaine, and P\v{a}tra\c{s}cu \cite{baran2005subquadratic} on subquadratic algorithms for 3SUM.
\end{abstract}

\thispagestyle{empty}
\tableofcontents
\newpage
\setcounter{page}{1}


\section{Introduction}
\label{sec:intro}

One of the most well-known and simple-to-state NP-complete problems is the \emph{Subset Sum} problem. An instance of Subset Sum consists of a list $X=(x_1,\dots,x_n)$ of $n$ positive integer values and a positive integer target $t$, and the output is either a subset $S \subseteq X$ such that $\sum_{x_i \in S} x_i = t$ or a report that no such subset exists. Subset Sum was one of the original 21 problems proved NP-complete in Karp's seminal paper \cite{Karp72} and has been the subject of intensive study from many different perspectives for at least five decades.

This paper is motivated by 
  the following open problem in the theory of exact exponential time algorithms: how quickly can we solve worst-case instances of Subset Sum? Exhaustive search over all possible solutions yields a trivial $2^n \cdot \poly(n)$-time algorithm. In 1974, Horowitz and Sahni introduced the ``meet-in-the-middle'' technique, which gives an algorithm that can be implemented in $O(2^{n/2})$ time in standard RAM models \cite{horowitz1974computing}. Since then, obtaining a $2^{(1/2 - c)n}$-time algorithm for some constant $c>0$ has emerged as a major goal in the exact exponential time algorithms community (explicitly mentioned in \cite{woeginger2008open,open-problem-FPT-school,austrin2015subset,austrin2016dense,NederlofW21} and numerous other works) which has attracted the attention of many researchers.  

Intriguing progress has been made on a number of variants of the core worst-case Subset Sum problem.  More than forty years ago Schroeppel and Shamir \cite{schroeppel1981t} improved the $2^{n/2} \cdot \poly(n)$ space complexity of the meet-in-the-middle algorithm by giving an algorithm that runs in $2^{n/2} \cdot \poly(n)$ time and $2^{n/4} \cdot \poly(n)$ space. An exciting recent breakthrough by Nederlof and Wegrzycki \cite{NederlofW21} further improved this space complexity to $2^{0.249999n}$.  In \cite{howgrave2010new} Howgrave-Graham and Joux gave an algorithm which can solve \emph{average-case} Subset Sum instances in time $2^{0.337n}$,\footnote{See the last paragraph of \cite{becker2011improved} for a discussion of the runtime of \cite{howgrave2010new}.} and this was later improved to $2^{0.291n}$ in the work of Becker et al.~\cite{becker2011improved}. The closely related \emph{Equal Subset Sum} problem, which looks for two subsets with the same sum, can be solved exponentially faster than suggested by meet-in-the-middle \cite{mucha2019equal} and also yields further improvements in the average case \cite{chen2022average}. However, to the best of our knowledge, there have been no improvements on the worst-case $O(2^{n/2})$ runtime of the meet-in-the-middle algorithm for Subset Sum since it was first introduced almost fifty years ago. 

\vspace{.1in}
\noindent
{\bf Our contribution: Worst-case Subset Sum in $2^{n/2}/\poly(n)$ time.}
Given the longstanding difficulty of achieving a $2^{(1/2 - c)n}$-time worst-case algorithm for Subset Sum, it is natural to consider the relaxed goal of achieving \emph{some} nontrivial speedup of the meet-in-the-middle algorithm. In this paper we achieve this goal; more precisely, we give three different randomized algorithms for worst-case Subset Sum, each of which runs in\ignore{expected} time $O(2^{n/2} \cdot n^{-\gamma})$ for a specific constant $\gamma > 0$ in a standard word RAM or circuit RAM model (described in detail in \Cref{sec:model} below).  Our fastest algorithm, which combines techniques from our other two algorithms, runs in time $O(2^{n/2} \cdot n^{-0.5023})$.

The improvements we achieve over the $O(2^{n/2})$ runtime of the meet-in-the-middle algorithm for Subset Sum are analogous to ``log-shaving'' improvements on the runtimes of well-known and simple polynomial-time algorithms for various problems which have resisted  attempts at polynomial-factor improvements.   There is a substantial strand of research along these lines (see \cite{Chan13,Chan-talk-logshaving}  for a non-exhaustive overview); indeed, Abboud and Bringmann \cite{AbboudBringmann18} have recently stated that: ``A noticeable fraction of Algorithms papers in the last few decades improve the runtime of well-known algorithms for fundamental problems by logarithmic factors.'' In our setting, since the well-known and simple algorithm for Subset Sum (namely, meet-in-the-middle)  runs in exponential time, saving a $\poly(n)$ factor, as we do, is analogous to ``log-shaving''.  Indeed, as we discuss in \Cref{sec:techniques} below, our first and most straightforward algorithm is based on ``bit-packing'' techniques that were used by Baran, Demaine, and P\v{a}tra\c{s}cu  \cite{baran2005subquadratic} to shave log factors from the standard $O(n^2)$-time algorithm for the 3SUM problem.  We find it  somewhat surprising that the ``log-shaving'' perspective has not previously appeared in the literature on Subset Sum, and we hope that our work will lead to further (and more substantial) runtime improvements for Subset Sum and other problems with well-known and simple exponential-time algorithms.

\begin{remark}
We note that by an easy reduction, an algorithm for 4SUM running in time $O(n^2 / \log(n)^\alpha)$ for any constant $\alpha > 0$ would immediately imply a Subset Sum algorithm running in time $O(2^{n/2} / n^\alpha)$, which would be a result comparable to ours. However, while log-shaving results for 3SUM are known \cite{baran2005subquadratic, gronlund2014threesomes}, giving an $o(n^2)$ algorithm for 4SUM is a well-known open problem.
\end{remark}

\subsection{Our Computational Model}
\label{sec:model}

Before presenting our results and techniques we describe the models of memory and computation that we use. 
A Subset Sum instance is parameterized by the number of inputs $n$ and the size of the target value $t$ (without loss of generality, $x_1, x_2, \dots, x_n \leq t$). Thus it is natural to adopt a memory model with word length $\word = \Theta(\log t)$ such that each input integer can be stored in a single word. This is the framework used in the work of Pisinger \cite{pisinger2003dynamic}, which studies dynamic programming approaches for Subset Sum in the word RAM model (see \cite[Equation~(1)]{pisinger2003dynamic}). We also note that this memory model is analogous to the standard RAM model that is commonly used for problems such as 3SUM (see e.g.~\cite{baran2005subquadratic}), where it is assumed that each input value is at most $\poly(n)$ and hence fits into a single $O(\log n)$-bit machine word.

This framework lets us consider arbitrary input instances of Subset Sum with no constraints on the size of the input integers. If $t = 2^{o(n)}$, standard dynamic programming algorithms \cite{bellman1966dynamic} solve the problem in time $O(nt) = 2^{o(n)}$, which supersedes our $\poly(n)$-factor improvements over meet-in-the-middle; hence throughout the paper we assume $t = 2^{\Omega(n)}$. It is arguably most natural to think about instances in which $t = 2^{\Theta(n)}$, in which case $\word = \Theta(n)$, and we encourage the first-time reader to imagine $\word = \Theta(n)$ for easy digestion. More precisely, we make the assumption throughout the paper that the word size $\word = \poly(n)$, although some of our results even hold for extremely large word sizes and are footnoted accordingly.

We consider runtime in two standard variants of the RAM model. The first is \emph{circuit RAM}; in this model, any operation that maps a constant number of words to a single word and has a $\poly(\word)$-size circuit with unbounded fan-in gates can be performed in time proportional to the depth of the circuit. Consequently, in the circuit RAM model, $\mathsf{AC}^0$ operations on a {\em constant} number of words can be performed in {\em constant} time, and multiplying, performing modular division, etc., on two $\word$-bit words can be performed in time $O(\log \word)$.
The second is \emph{word RAM}, in which the usual arithmetic operations, including multiplication, are assumed to take unit time, but arbitrary $\mathsf{AC}^0$ operations are not atomic operations on words. We present each of our algorithms for the stronger circuit RAM model,\footnote{Note that any algorithm in the word RAM model can be simulated in the circuit RAM model with no more than an $O(\log \word)$-factor slowdown.} and explain adaptations that give corresponding word RAM algorithms. 

\subsection{Results, Techniques, and Organization}
\label{sec:techniques}

In \Cref{sec:prelim} we establish our notation and review some background results and observations that will be used throughout the paper.

\Cref{sec:bit-packing,sec:memo-ov,sec:main} give our three new algorithms, which augment the standard meet-in-the-middle approach in different ways to achieve their respective runtime improvements.
Each of our algorithms is a randomized decision algorithm that runs in the time bound claimed below and on every input instance outputs the correct answer with probability at least $3/4$.
Further, each of our algorithms has one-sided error, i.e., it never makes a mistake when it outputs ``yes''. (See \Cref{obs:search-decision} for a discussion of how such randomized  decision algorithms yield randomized search algorithms with the same asymptotic runtime.)

Our first and simplest algorithm, presented in \Cref{sec:bit-packing}, achieves a runtime of $\tO(2^{n/2} \cdot \word^{-1/2}) \leq \tO(2^{n/2} \cdot n^{-1/2})$ in the circuit RAM model and $\tO(2^{n/2} \cdot n^{-1/2})$ in the word RAM model, for all $\word=\poly(n)$.\footnote{The notation $\tO(\cdot)$ suppresses $\polylog(\word) = \polylog(n)$ factors.} It works by adapting the {\em bit-packing} trick, a technique developed by Baran, Demaine, and P\v{a}tra\c{s}cu \cite{baran2005subquadratic} for the 3SUM problem, for the \blackref{alg:MeetInTheMiddle} algorithm.
The idea is to compress the two lists of partial subset sums used in \blackref{alg:MeetInTheMiddle} by packing hashes of multiple values into a machine word, while preserving enough information to make it possible to run (an adaptation of) \blackref{alg:MeetInTheMiddle} on the lists of hashed and packed values. This results in a runtime savings over performing \blackref{alg:MeetInTheMiddle} on the original lists (without hashing and packing), because processing a pair of words, each containing multiple hashed values, takes constant expected time in the circuit RAM model and can be memoized to take constant time in the word RAM model.

Our second algorithm, given in \Cref{sec:memo-ov}, achieves a runtime of $O(2^{n/2} \cdot \word^{-\gamma}) \leq O(2^{n/2} \cdot n^{-\gamma})$ for some constant $\gamma > 0.01$ in the circuit RAM model and $O(2^{n/2} \cdot n^{-\gamma})$ in the word RAM model, for all $\word = \poly(n)$. Although the time savings is smaller than our first algorithm, we believe that this algorithm is conceptually interesting since it avoids bit-packing and instead combines \blackref{alg:MeetInTheMiddle} with two techniques devised in prior work on Subset Sum. The first of these is the ``representation method'' introduced by Howgrave-Graham and Joux \cite{howgrave2010new}. Roughly speaking, the idea of this method is to (i) increase the size of the search space in such a way that a single solution has many ``representations'' in the space of enhanced solutions, and then (ii) search over only a fraction of the enhanced solution space. A consequence of expanding the solution space, though, is that a number of ``pseudosolutions'', solutions that contain certain input elements more than once, are introduced. This leads us to the second technique, i.e.\ the use of a fast subroutine for the Orthogonal Vectors (OV) problem (recall that OV is the problem of deciding whether two lists of $\{0,1\}$-vectors contain a pair of vectors, one from each list, that are orthogonal). The fast OV subroutine lets us efficiently rule out pseudosolutions while running (an adaptation of) \blackref{alg:MeetInTheMiddle} on a fraction of the enhanced solution space.

In \Cref{sec:main} we give our fastest algorithm, which uses a delicate combination of the techniques from \Cref{sec:bit-packing,sec:memo-ov} to obtain a runtime of $O(2^{n/2} \cdot n^{-0.5023})$ for all $\word=\poly(n)$.
While the runtime improvement over \Cref{sec:bit-packing} is not large, this algorithm demonstrates that by leveraging insights specific to the Subset Sum problem, we can achieve time savings beyond what is possible with more ``generic'' log shaving techniques.
Finally, in \Cref{sec:extensions} we briefly discuss directions for future work.

\section{Preliminaries}
\label{sec:prelim}

To ease readability, we adopt the following notational conventions throughout the paper:
lowercase Roman letters ($\word$, $n$, etc.) denote variables; lowercase Greek letters ($\epsilon$, $\alpha$, etc.) denote numerical constants; capital Roman letters ($L$, $W$, etc.) denote sets, multisets, or lists; and calligraphic capital letters ($\calW$, $\calQ$, etc.) denote collections of sets of numbers.

\vspace{.1in}
\noindent
\textbf{Logarithms.}
When written without a specified base, $\log(\cdot)$ denotes the base-2 logarithm.

\vspace{.1in}
\noindent
\textbf{Big-O Notation.}
We augment big-$O$ notation with a tilde ($\tO$, $\tOmega$, $\tTheta$ etc.) to suppress $\polylog(\word) = \polylog(n)$ factors. For example, we have $\tilde{O}(2^n) = O(2^n \cdot \polylog(n))$ and $\tilde{\Omega}(n) = \Omega(\frac{n}{\polylog(n)})$.

\vspace{.1in}
\noindent
\textbf{Probability Notation.}
Random variables are written in boldface. In particular, we write  ``$\bx \sim S$'' to indicate that an element $\bx$ is sampled uniformly at random from a finite multiset $S$.

\vspace{.1in}
\noindent
\textbf{Set Notation.}
We write $[a:b]$ for the set of integers $\{a, a+1, \dots, b\}$ and $[a]$ for $\{1, 2, \dots, a\}$. We write $\P[a: b]$ for the set of all primes in the interval $[a:b]$.

Given a multiset or list $Y$ of integers, we adopt several shorthands: $\Sigma(Y) := \sum_{y \in Y} y$ denotes the sum,\ignore{$\calW(Y) := \{T \mid T \subseteq Y\}$ denotes the collection of sub-multisets,}
$W(Y) := \{\Sigma(T) \mid T \subseteq Y\}$ denotes the set of {\em distinct} sub-multiset sums, and $L_Y$ denotes the list containing the elements of $W(Y)$ sorted in increasing order.

Also, we often write $f(Y)$ for the multiset or list obtained by applying operation $f$ element-wise, such as $Y + \alpha = \{y + \alpha \mid y \in Y\}$ and $\alpha Y = \{\alpha y \mid y \in Y\}$.
The only exception is $(Y \bmod p)$, which denotes the set of {\em distinct} residues $\{(y \bmod p) \mid y \in Y\}$.

\vspace{.1in}
\noindent
\textbf{Stirling's Approximation for Binomial Coefficients.}
We use the following well-known consequence of Stirling's approximation (see e.g.\ \cite[Lemma~16.19]{FG06}): For each integer $j \in [0: n / 2]$,
\begin{equation}
    \label{eq:stirling}
    \sum_{i \,\leq\, j} \binom{n}{i} ~\leq~ 2^{H(j / n) \cdot n},
\end{equation}
where $H(y) := -y \log(y) - (1 - y) \log(1 - y)$ is the binary entropy function.

\vspace{.1in}
\noindent
\textbf{Pseudolinear Hashing.}
Recall that $\word = \poly(n)$ is the word length in our memory model.
Given an integer $m \leq \word$, we write $\bh_m$ to denote the random hash function defined as
\begin{equation} \label{eq:random-hash}
    \bh_m(y) ~~:=~~ (\bu \cdot y \pmod{2^\word}) \gg \word - m.
\end{equation}
Here the input $y$ is an $\word$-bit integer, $\bu$ is selected uniformly at random from all odd $\word$-bit integers, and $\gg$ denotes a bit shift to the right, i.e., dividing $\bu \cdot y \pmod{2^\word}$ by $2^{\word - m}$ and then truncating the result so that only the higher-order $m$ bits remain.

This hash function $\bh_{m}(y)$ can be evaluated in time $O(\log \word) = O(\log n)$ for $\word = \poly(n)$ in the circuit RAM model, i.e., essentially the time to multiply, or constant time $O_{n}(1)$ in the word RAM model. Further, $\bh_{m}$ has the following useful properties \cite{dietzfelbinger1997reliable,baran2005subquadratic}.

\begin{lemma}[Pseudolinear Hashing \cite{dietzfelbinger1997reliable,baran2005subquadratic}]
\label{lem:hash}
The following hold for the hash function $\bh_{m}$:
\begin{enumerate}
    \item \textbf{Pseudolinearity.} For any two $\word$-bit integers $y$, $z$ and any outcome of $\bh_m$,
    \[
        \bh_{m}(y) + \bh_{m}(z) ~\in~ \bh_{m}(y + z) - \{0, 1\} \pmod{2^m}.
    \]
    \item \textbf{Pseudouniversality.} For any two $\word$-bit integers $y$, $z$ with $y \neq z$,
    \[
        \Prx \big[\bh_{m}(y) = \bh_{m}(z)\big] ~=~ O(2^{-m}).
    \]
\end{enumerate}
\end{lemma}

\subsection{Preliminary Results}

First of all, for the purposes of this paper, we can assume that
the input target $t$ have size $2^{\Omega(n)}$:

\begin{observation}[Assumptions about Input Instances]
\label{obs:parameters}
Given a Subset Sum instance $(X,\, t)$ with $t \leq 2^{0.499n}$, the standard dynamic programming algorithm \cite{bellman1966dynamic} takes time $O(nt) = 2^{0.499n + o(n)}$, much faster than the $\poly(n)$-factor speedups we are targeting. Hence without loss of generality, we assume $t = 2^{\Omega(n)}$ and $\word = \Omega(\log t) = \Omega(n)$ throughout the paper.
\end{observation}

We further observe that the randomized {\em decision} algorithms that we give can be converted into randomized {\em search} algorithms with the same asymptotic runtimes:

\begin{observation}[Search-to-Decision Reductions]
\label{obs:search-decision}
Given a deterministic decision algorithm \texttt{DET} for the Subset Sum problem, the standard search-to-decision reduction solves an $n$-integer instance by running \texttt{DET} a total of $n$ times, with the $i$-th execution being on an $(n-i+1)$-integer subinstance. Each of our randomized algorithms \texttt{RAN} succeeds with probability $\geq 3/4$, returns no false positives, and runs in time $\mathrm{TIME}(n) = 2^{\Theta(n)}$.
On the same $(n-i+1)$-integer subinstance, we can reduce the failure probability to $\leq 1/4^{i}$ by repeating \texttt{RAN} a number of $i$ times.
Overall, the success probability is $\geq 1 - \sum_{i=1}^n 1/4^{i} \geq 2/3$ and the runtime is $\sum_{i=1}^n i \cdot \mathrm{TIME}(n+1-i) = O(\mathrm{TIME}(n))$.
\end{observation}

A basic primitive for our algorithms is the sorted list $L_{Y}$ containing the elements of $W(Y)$, all {\em distinct} subset sums, for a given multiset $Y$ of size $|Y| = k$. \Cref{fig:alg:SSE} shows an $O(2^k)$-time folklore algorithm that enumerates $L_Y$:

\begin{figure}[t]
    \centering
    
    \begin{mdframed}
    Procedure $\term[\SSE]{alg:SSE}(Y)$

    \begin{flushleft}
    {\bf Input:} An integer multiset $Y = \{y_1, \dots, y_k\}$.
    
    \vspace{.05in}
    {\bf Output:} A sorted list $L_{Y} := L_{k}$ of $W(Y)$, all distinct subset sums for $Y$.
    
    \begin{enumerate}
        \item Initialize $L_0 := (0)$. Then for each $i \in [k]$:
        
        \item \qquad Create the sorted list $L'_{i-1} := L_{i-1} + y_i$.
        
        \item \qquad Create the sorted list $L_i$ obtained by merging $L_{i-1}$, $L'_{i-1}$ and removing duplicates.
    \end{enumerate}
    \end{flushleft}
    \end{mdframed}
    \caption{Efficiently enumerating subset sums of $k \geq 1$ integers.}
    \label{fig:alg:SSE}
\end{figure}

\begin{lemma}[Sorted Sum Enumeration; Folklore]
\label{lem:sse}
\blackref{alg:SSE} runs in time $O(2^k)$.
\end{lemma}

\begin{proof}
We essentially adapt the classic {\em merge sort} algorithm.
Since $|L_{i - 1}| = |L'_{i - 1}| \leq |L_{i}| \leq 2|L_{i - 1}|$ for each $i \in [k]$, the runtime of \blackref{alg:SSE} is
\[
    \textstyle
    \sum_{i \in [k]} O(|L'_{i - 1}| + |L_{i}|)
    ~=~ \sum_{i \in [k]} O(|L_{i}|)
    ~=~ \sum_{i \in [k]} O(2^{i})
    ~=~ O(2^{k}).
    \qedhere
\]
\end{proof}

We can improve this runtime analysis if $|W(Y)|$ is smaller than $2^k$ by a $\poly(k)$ factor:

\begin{lemma}[Sorted Sum Enumeration for Small $W(Y)$]
\label{lemma:sse-poly}
If $Y$ is a multiset of $|Y| = k$ integers with $|W(Y)| \leq 2^k \cdot k^{-\epsilon}$ for some constant $\epsilon > 0$, \blackref{alg:SSE} runs in time $O(2^k \cdot k^{-\epsilon} \cdot \log k)$.
\end{lemma}
    
\begin{proof}
If $|W(Y)| \leq 2^k \cdot k^{-(1+\epsilon)}$, then since $|L_i| \leq |L_k| = |W(Y)|$ for each $i \in [k]$, it is easy to check that the algorithm runs in time $O(k \cdot |W(Y)|) = O(2^k \cdot k^{-\epsilon})$;
so suppose that $|W(Y)| \geq 2^k \cdot k^{-(1+\epsilon)}$. Using the bound $|L_i| \leq 2^i$ for $i \leq \log |W(Y)|$ and the bound $|L_i| \leq |W(Y)|$ for $\log |W(Y)| < i \leq k$, the runtime of \blackref{alg:SSE} is upper bounded by
\begin{align*}
    \textstyle\sum_{i \in [k]} O(|L_{i}|)
    & ~=~ \textstyle\sum_{i \,\leq\, \log|W(Y)|} O(2^i) + \sum_{\log|W(Y)| < i \leq k} O(|W(Y)|) \\
    & ~=~ O(|W(Y)|) + (1+\epsilon) \cdot \log k \cdot O(|W(Y)|) \\
    & ~=~ O(|W(Y)| \cdot \log k)
    ~=~ O(2^k \cdot k^{-\epsilon} \cdot \log k).
    \qedhere
\end{align*}
\end{proof}

\begin{remark}
\Cref{lemma:sse-poly} gives a tight analysis of the algorithm when $|W(Y)| = 2^{k} \cdot k^{-\epsilon}$, as can be seen by considering the particular size-$k$ multiset $Y = (2^0, 2^1, \dots, 2^{k - \epsilon \log k - 1}, 1, 1, \dots, 1)$.
\end{remark}

\noindent

Finally, \Cref{fig:alg:MeetInTheMiddle} shows the classic meet-in-the-middle algorithm for Subset Sum, by Horowitz and Sahni \cite{horowitz1974computing}, which serves as our baseline for comparison. 
\begin{lemma}
\label{mim-runtime}
    The worst-case runtime of \blackref{alg:MeetInTheMiddle} is $O(2^{n/2}).$
\end{lemma}
\begin{proof}
    This follows immediately from \Cref{lem:sse} as $|L_A|, |L_B| \leq 2^{n/2}$.
\end{proof}

\begin{figure}[t]
    \centering
    \begin{mdframed}
    Procedure $\term[\MeetInTheMiddle]{alg:MeetInTheMiddle}(X,\, t)$

    \begin{flushleft}
    {\bf Input:} An integer multiset $X = \{x_{1},\, x_{2},\, \dots,\, x_{n}\}$ and an integer target $t$.
    
    \vspace{.05in}
    {\bf Output:} ``yes'' if $(X,t)$ is a Subset Sum instance that has a solution, ``no'' otherwise.

    \begin{enumerate}
    \setcounter{enumi}{-1}
        \item\label{alg:MeetInTheMiddle:1}
        Fix any partition of $X = A \cup B$ such that $|A| = |B| = {n}/{2}$.
        
        \item\label{alg:MeetInTheMiddle:2}
        Enumerate the sorted lists $L_A$ and $L_B$ using \blackref{alg:SSE}.
        
        \item\label{alg:MeetInTheMiddle:3}
        Initialize two pointers at the smallest value in $L_A$ and the largest value in $L_B$. \\
        If these two values sum to the target $t$, then return ``yes''; \\
        if they sum to less than $t$, then increment the pointer into $L_A$ and repeat; \\
        and if they sum to more than $t$, then increment the pointer into $L_B$ and repeat. \\
        If either pointer goes past the end of its list, then return ``no''.

    \end{enumerate}
    \end{flushleft}
    \end{mdframed}
    \caption{The classic {\MeetInTheMiddle} algorithm \cite{horowitz1974computing}.
    \label{fig:alg:MeetInTheMiddle}}
\end{figure}


\section{\texorpdfstring{$\Omega(n^{0.5} / \log n)$}{}-Factor Speedup via Bit Packing}
\label{sec:bit-packing}

In this section we analyze our first and simplest algorithm, \blackref{alg:BitPacking} (see \Cref{fig:alg:BitPacking}).\footnote{As explained in \Cref{sec:model}, we assume $\word = \poly(n)$ throughout; however, our results for \blackref{alg:BitPacking} apply for superpolynomial word length as long as $\word = O(2^{n/3})$. We leave the extensional modifications to the proofs for this regime as an exercise for the interested reader.}

\begin{figure}[t]
    \centering
    \begin{mdframed}
    Procedure $\term[\BitPacking]{alg:BitPacking}(X,\, t)$

    \begin{flushleft}
    {\bf Input:} An integer multiset $X = \{x_{1},\, x_{2},\, \dots,\, x_{n}\}$ and an integer target $t$.
    
    \vspace{.05in}
    {\bf Setup:} Draw a random hash function $\bh_m$ with $m = 3\log\word$ (cf.~\Cref{eq:random-hash}).
    
    \begin{enumerate}
    \setcounter{enumi}{-1}
        \item Fix any partition of $X = A \cup B \cup D$ such that $|D| = \log \word$ and $|A| = |B| = \frac{n - |D|}{2}$.
        
        \item\label{alg:BitPacking:1}
        Create the set
        $W(D)$ and the sorted lists $L_{A}$, $L_{B}$ using \blackref{alg:SSE}.
        
        \item\label{alg:BitPacking:2} 
        Create lists $\bh_{m}(L_{A})$ and $\bh_{m}(L_{B})$ by applying $\bh_m$ element-wise to $L_{A}$ and $L_{B}$. \\
        Let $\bH_{A}$ be the list obtained from $\bh_{m}(L_{A})$ by packing $(\word / m)$ elements of $\bh_m(L_{A})$ \\
        into each $\word$-bit word of $\bH_{A}$ (preserving the sorted ordering) and likewise for $\bH_{B}$.
        
        \item\label{alg:BitPacking:3}
        For each $t' \in (t - W(D))$:
        
        \item\label{alg:BitPacking:4}
        \qquad Initialize indices $i := 0$ and $j := |\bH_{B}|-1$. While $i < |\bH_{A}|$ and $j \geq 0$:
        
        \item\label{alg:BitPacking:5}
        \qquad\qquad If the indexed words $\bH_{A}[i]$ and $\bH_{B}[j]$ contain a pair of hashes $(\bh_{m}(a'),\, \bh_{m}(b'))$ \\
        \qquad\qquad such that $\bh_{m}(a') + \bh_{m}(b') \in \bh_{m}(t') - \{0,\, 1 \} \pmod{2^{m}}$, \\
        \qquad\qquad use \blackref{alg:MeetInTheMiddle} to search for a solution \\
        \qquad\qquad $(a,\, b) \in L_{A}[i\word / m: (i+1)\word / m - 1] \times L_{B}[j\word/m: (j+1)\word / m -1]$ \\
        \qquad\qquad such that $a + b = t'$. Halt and return ``yes'' if a solution is found.
        
        \item\label{alg:BitPacking:6}
        \qquad\qquad If $L_{A}[(i+1)\word / m - 1] + L_{B}[j\word / m] < t'$, increment $i \gets i + 1$. \\
        \qquad\qquad Otherwise, decrement $j \gets j - 1$.
        
        \item\label{alg:BitPacking:7}
        Return ``no'' (i.e., no solution was found for any $t' \in (t - W(D))$).
    \end{enumerate}
    \end{flushleft}
    \end{mdframed}
    \caption{The {\BitPacking} algorithm.
    \label{fig:alg:BitPacking}}
\end{figure}

\begin{theorem}
\label{thm:bit-packing}
\blackref{alg:BitPacking} is a zero-error randomized algorithm for the Subset Sum problem with expected runtime $O(2^{n / 2} \cdot \word^{-1 / 2} \cdot \log \word) \leq O(2^{n / 2} \cdot n^{-1 / 2} \cdot \log n)$ in the circuit RAM model.\footnote{By halting \blackref{alg:BitPacking} and returning ``no'' if its runtime exceeds $C \cdot 2^{n/2} \cdot \word^{-1/2} \cdot \log \word$ for a large enough constant $C > 0$, we get an one-sided error algorithm with success probability $\geq 3/4$, as claimed in \Cref{sec:techniques}.}
\end{theorem}

\blackref{alg:BitPacking} works by packing $\word/m$ hashed values into a single word via our pseudolinear hash function $\bh_{m}$, for $m = 3\log\word$, while preserving enough information to run \blackref{alg:MeetInTheMiddle} on the lists of hashed and packed values. This allows us to compare two length-$(\word / m)$ sublists of $L_{A}$ and $L_{B}$ in constant expected time in the circuit RAM model, since each hashed and packed sublist fits into a constant number of words, instead of time $O(\word / m)$ like the original \blackref{alg:MeetInTheMiddle} algorithm.
So far, this is essentially the approach taken by \cite{baran2005subquadratic} in their bit-packing algorithm for 3SUM. However, in our context the $O(\word / m)$ speedup described above is offset by the following issue: if we follow the original \blackref{alg:MeetInTheMiddle} setup and take $|A| = |B| = \frac{n}{2}$, the lists $L_{A}$ and $L_{B}$ may have length $\Omega(2^{n/2})$, increasing the runtime.

To deal with this, we set aside a small set $D \subseteq X$ of $|D| = \log \word$ many input elements and solve the remaining subinstance $X \setminus D$ for each shifted target $t' \in (t - W(D))$. Removing the elements in $D$ shortens the lists $L_A$ and $L_B$, which are now formed from the elements in $X \setminus D$, and allows us to enumerate and pack them quickly.
Balancing the overhead of solving each of these subinstances against the savings described earlier, we get the claimed speedup $\Omega(\word^{1/2} / \log \word) \geq \Omega(n^{1/2} / \log n)$.

\begin{proof}[Proof of Correctness for \blackref{alg:BitPacking}]
The algorithm outputs ``yes'' only if a triple $(a,\, b,\, t') \in L_{A} \times L_{B} \times (t - W(D))$ with $a + b = t'$ is found, so it never returns a false positive.

It remains to show that for any $t' \in (t - W(D))$, we are guaranteed to find a shifted solution $(a,\, b) \in L_{A} \times L_{B}$ such that $a + b = t'$, if one exists.
Without loss of generality, we consider two sublists $L_{A}[i\word / m: (i+1)\word / m - 1]$ and $L_{B}[j\word / m: (j+1)\word / m -1]$ that contain such a shifted solution $(a,\, b)$ and correspond to two packed words $\bH_{A}[i]$ and $\bH_{B}[j]$ for some indices $i$ and $j$.
The existence of such a shifted solution $a + b = t'$ combined with the condition in Line~\ref{alg:BitPacking:6} ensures that the algorithm will not step past either the packed word $\bH_{A}[i]$ or $\bH_{B}[j]$ before reaching the other one, so the algorithm will compare these two packed words at some point.
Following \Cref{lem:hash}, we have $\bh_{m}(a) + \bh_{m}(b) \in \bh_{m}(t') - \{0,\, 1 \} \pmod{2^{m}}$, satisfying the condition in Line~\ref{alg:BitPacking:5}.
Thus we are guaranteed to find the shifted solution $(a,\, b)$ by running \blackref{alg:MeetInTheMiddle} to check all pairs in $L_{A}[i\word / m: (i+1)\word / m - 1] \times L_{B}[j\word/m: (j+1)\word / m -1]$.
\end{proof}

\begin{proof}[Proof of Runtime for \blackref{alg:BitPacking}]
Recalling the assumption $\word = \poly(n)$:
\begin{flushleft}
\begin{itemize}
    \item Line~\ref{alg:BitPacking:1} takes time $O(2^{|A|} + 2^{|B|} + 2^{|D|}) = O(2^{n / 2} \cdot \word^{-1 / 2} + \word) = O(2^{n / 2} \cdot \word^{-1 / 2})$ by \Cref{lem:sse}, for the choices of $|A|$, $|B|$, and $|D|$.
    
    \item Line~\ref{alg:BitPacking:2} takes time $(|L_{A}| + |L_{B}|) \cdot O(\log \word) = O(2^{n / 2} \cdot \word^{-1 / 2} \cdot \log \word)$, where $O(\log \word)$ bounds the time for each evaluation of the hash function $\bh_m$.\ignore{the time to multiply two words.}
    
    \item Line~\ref{alg:BitPacking:3} (the outer loop) is performed for at most $|W(D)| \leq 2^{|D|} = \word$ iterations.
    
    \item Line~\ref{alg:BitPacking:4} (the inner loop) is performed for at most $(|L_{A}| + |L_{B}|) \cdot \frac{1}{\word / m} = O(2^{n / 2} \cdot \word^{-3 / 2} \cdot \log \word)$ iterations, since each iteration either increments $i \gets i + 1$ or decrements $j \gets j - 1$.
    
    \item Line~\ref{alg:BitPacking:5}:
    (i)~Checking whether the ``If'' condition holds for any two words $\bH_{A}[i]$ and $\bH_{B}[j]$ requires a single $\mathsf{AC}^0$ operation on three words, taking constant time in the circuit RAM model. (ii)~Finding a solution $(a,\, b)$ using \blackref{alg:MeetInTheMiddle} on the two length-$(\word/m)$ sublists takes time $O(\word / m) = O(\word / \log \word)$.
    
    The ``If'' test is passed (i)~at most once for a correct solution and (ii)~each time we encounter a hash collision. 
    Note that the sequence of pairs of words $(\bH_A[i],\, \bH_B[j])$ we compare is completely determined by $L_A$ and $L_B$ and is unaffected by the outcome of the random hash function $\bh_{m}$.
    Thus by \Cref{lem:hash}, each of the $(\word / m)^{2}$ hash pairs in $(\bH_A[i],\, \bH_B[j])$ incurs a collision with probability $O(2^{-m})$. By a union bound, the expected time taken for Line~\ref{alg:BitPacking:5} because of hash collisions is at most $(\word / m)^{2} \cdot O(2^{-m}) \cdot O(\word / m) = O(1 / \log^{3}(\word)) = o_n(1)$, since $m = 3\log \word$ and $\word = \Omega(n)$.
    
    \item Line~\ref{alg:BitPacking:6} clearly takes time $O_{n}(1)$.
\end{itemize}
\end{flushleft}
Consequently, \blackref{alg:BitPacking} has expected runtime
\begin{align*}
    \mathrm{TIME}(n,\, \word)
    & ~=~ O(2^{n / 2} \cdot \word^{-1 / 2} + \word) + O(2^{n / 2} \cdot \word^{-1 / 2} \cdot \log \word)
    && \text{Lines~\ref{alg:BitPacking:1} and \ref{alg:BitPacking:2}} \\
    & \phantom{~=~} + 1 \cdot \big(O_{n}(1) + O(\word / \log \word) + O_{n}(1)\big)
    && \text{Lines~\ref{alg:BitPacking:3} to \ref{alg:BitPacking:6}; solution} \\
    & \phantom{~=~} + \word \cdot O(2^{n / 2} \cdot \word^{-3 / 2} \cdot \log \word) \cdot \big(O_{n}(1) + o_{n}(1) + O_{n}(1)\big)
    && \text{Lines~\ref{alg:BitPacking:3} to \ref{alg:BitPacking:6}; collisions} \\
    & ~=~ O(2^{n / 2} \cdot \word^{-1 / 2} \cdot \log \word). && \qedhere
\end{align*}
\end{proof}

\begin{observation}[Adapting \blackref{alg:BitPacking} to Word RAM]
\label{obs:bit-packing-word-RAM}
In the word RAM model, multiplication and evaluation of our pseudolinear hash function $\bh_{m}$ each take constant time. Hence, for any word length $\word = \Omega(n)$ we can get a variant of \blackref{alg:BitPacking} with expected runtime $O(2^{n/2} \cdot n^{-1/2} \cdot \log n)$, essentially by performing \blackref{alg:BitPacking} as if the word length were $\word' := 0.1n$. By a similar conversion from \cite{baran2005subquadratic}:

\vspace{.05in}
\noindent
Run \blackref{alg:BitPacking} as if the word length were $\word' = 0.1n$, which results in the modified parameters $m' = 3\log \word'$, $|D'| = \log\word'$, and $|A'| = |B'| = (n - |D'|) / 2$ etc., except for two modifications:
\begin{enumerate}
    \item Line~\ref{alg:BitPacking:2} packs $q' := \min\{\word',\, \word\} / m' = \Theta(\frac{n}{\log n})$ many $m'$-bit hashes into each word of $\bH_{A'}$, $\bH_{B'}$, so every $(\word' / m')$ hashes are stored in $\word' / (m' q') = \Theta_n(1)$ words rather than a single word.
    
    \item Before Line~\ref{alg:BitPacking:5}, create a table that memoizes the result of every comparison of two $\word'$-bit strings in time $(2^{\word'})^2 \cdot \poly(\word') = O(2^{0.21n})$. This table can then be accessed via a $2\word' / (m' q') = \Theta_n(1)$-word index in constant time. Line~\ref{alg:BitPacking:5} replaces the constant-time $\mathsf{AC}^0$ circuit RAM operation on two $\word'$-bit strings with a constant-time lookup into this table.
\end{enumerate}

\vspace{.05in}
\noindent
Compared with running \blackref{alg:BitPacking} itself for $\word' = 0.1n$, the only difference of this variant is that $\bH_{A'}$ and $\bH_{B'}$ are stored in $\Theta_n(1)$ times as many words, so the correctness is easy to check. The expected time taken for the collisions in each execution of Line~\ref{alg:BitPacking:5} is $(q')^{2} \cdot O(2^{-m'}) \cdot O(q') = o_n(1)$. Thus, the overall runtime is as claimed:
\begin{align*}
    \underbrace{O(2^{|A'|} + 2^{|B'|} + 2^{|D'|} + 2^{0.21n})}_{\text{Lines~\ref{alg:BitPacking:1} and \ref{alg:BitPacking:2}}}
    ~+~ \underbrace{O(q') + 2^{|D'|} \cdot O((2^{|A'|} + 2^{|B'|}) / q')}_{\text{Lines~\ref{alg:BitPacking:3} to \ref{alg:BitPacking:6}; solution versus collisions}}
    ~=~ O(2^{n/2} \cdot n^{-1/2} \cdot \log n).
\end{align*}
\end{observation}

\newcommand{\bcalQ}{\boldsymbol{\calQ}}
\newcommand{\ResidueCoupleList}{\text{\tt Residue-Couple-List}}

\section{\texorpdfstring{$\Omega(n^{0.01})$}{}-Factor Speedup via Orthogonal Vectors and the Representation Method}
\label{sec:memo-ov}

Our second algorithm, \blackref{alg:memo-ov} (see \Cref{fig:alg:memo-ov}), achieves a speedup of $\Omega(\word^\gamma) \geq \Omega(n^{\gamma})$ over \blackref{alg:MeetInTheMiddle} for a constant $\gamma > 0.01$.\footnote{Similar to \Cref{sec:bit-packing}, while we investigate \blackref{alg:memo-ov} in the regime $\word = \poly(n)$, analogous results apply for word length $\word$ as large as $O(2^{c n})$ for an absolute constant $c > 0$.} 
While this is a smaller speedup than that achieved by the \blackref{alg:BitPacking} algorithm, the \blackref{alg:memo-ov} algorithm does not use ``bit tricks''.
Instead, \blackref{alg:memo-ov} combines \blackref{alg:MeetInTheMiddle} with the {\em representation method} of Howgrave-Graham and Joux \cite{howgrave2010new}, so as to reduce the Subset Sum problem to many small instances of the Orthogonal Vectors (OV) problem: namely, instances with $O(\word/ \log\word)$ many binary vectors of dimension $\Theta(\log \word) = \Theta(\log n)$. Such instances of OV can be solved quickly through a single $\mathsf{AC}^0$ word operation in the circuit RAM model (or through constantly many operations in the word RAM model after an initial memoization step), which leads to our speedup.

The high-level idea behind our algorithm is to partition the input $X = A \cup B \cup C$ into two large subsets\footnote{Technically, sub-multisets. We make the same simplification hereafter.} $A$ and $B$ and one small subset $C$ of size $|C| = \Theta(\log\word)$, and to run \blackref{alg:MeetInTheMiddle} on two lists formed from subsets of $(A \cup C)$ and $(B \cup C)$. We describe in more detail below just how these lists are formed, but roughly speaking they are created by modifying the representation method (in a way somewhat similar to the algorithm of Nederlof and Wegryzcki \cite{NederlofW21}) to ensure that the lists are not too long. Further, to eliminate the false positives due to overlapping subsets of $C$, we exploit a fast implementation of a function that computes a batch of small instances of Orthogonal Vectors. (To see the relevance of the Orthogonal Vectors problem in this context, note that a solution to an instance of OV on $k$-dimensional Boolean vectors corresponds to two disjoint subsets of the set $[k]$.)
Before giving more details we provide some helpful notation:

\vspace{.1in}
\noindent
\textbf{Notation and setup.}
We write $\calQ(C) := \{T \mid T \subseteq C ~\text{and}~ |T| \leq \frac{|C|}{4}\}$ to denote the collection of all quartersets for $C$. While $\calQ(C)$ is useful for intuition, in fact, as explained below, our algorithm will use a slight variant of it: we define $\calQ^{+\epsilon_2}(C)$ to be the collection of all subsets of size at most $(1+\epsilon_2)\frac{|C|}{4}$, where $\epsilon_2 > 0$ is a small constant that is fixed in the detailed description of the algorithm.

We denote by \term[\texttt{OV}]{function:OV} the Boolean function on $2\word = \poly(n)$ many input bits that takes as input two lists $(x^1, \dots, x^{\word / |C|})$, $(y^1, \dots, y^{\word / |C|})$ of (at most) $\word / |C|$ many binary vectors each, where each binary vector $x^i,\, y^j \in \{0, 1\}^{|C|}$, and returns $1$ if and only if the two lists contain an orthogonal pair, i.e., a pair $(i,\, j)$ such that $x^i_k \cdot y^j_k = 0$ for every $1 \leq k \leq |C|$. It is easy to see that \blackref{function:OV} is an $\mathsf{AC}^0$ operation on two $\word$-bit words, and thus it takes constant time in the circuit RAM model.

\vskip 0.1in
We return to the intuitive overview of our approach.
At a high level, the representation method works by first expanding the search space of possible solutions; for our algorithm this is done by writing down the list $\calQ^{+\epsilon_2}(C)$ of all ``near-quartersets''. If a certain Subset Sum solution $S \subseteq X$ satisfies $|S \cap C| \leq \frac{|C|}{2}$, the restricted solution $S \cap C = Q_1 \cup Q_2$ is the union of {\em many} different pairs of disjoint 
quartersets, namely $Q_1, Q_2 \in \calQ(C)$ with $Q_1 \cap Q_2 = \emptyset$.
In fact, we work on the list of near-quartersets, $\calQ^{+\epsilon_2}(C)$, rather than $\calQ(C)$, so as to cover all disjoint pairs $(Q_1,\, Q_2)$ with $|Q_1| + |Q_2| \approx |C|/2$.

We then \emph{filter} the list $\calQ^{+\epsilon_2}(C)$: given a random prime modulus $\bp$ we extract those near-quartersets that fall into two particular residue classes that sum to $\Sigma(S \cap C) \pmod{\bp}$. With appropriate preprocessing checks and parameter settings, this significantly reduces the search space while ensuring that we retain some disjoint pair of near-quartersets $Q_1,\, Q_2 \in \calQ^{+\epsilon_2}(C)$ that recover the restricted solution, $Q_1 \cup Q_2 = S \cap C$ with $Q_1 \cap Q_2 = \emptyset$.

Finally, we use a modified \blackref{alg:MeetInTheMiddle} procedure to search for a solution, i.e., two sum-subset couples $(a,\, Q_1) \in (L_{A} \times \calQ^{+\epsilon_2}(C))$, $(b,\, Q_2) \times (L_{B} \times \calQ^{+\epsilon_2}(C))$
with $a + \Sigma(Q_1) + b + \Sigma(Q_2) = \Sigma(S \cap A) + \Sigma(S \cap B) + \Sigma(S \cap C) = \Sigma(S) = t$, verifying $Q_1 \cap Q_2 = \emptyset$ via the Boolean function \blackref{function:OV}.

\begin{figure}[t]
    \centering
    \begin{mdframed}
    Procedure $\term[\RepresentationOV]{alg:memo-ov}(X,\, t)$

    \begin{flushleft}
    {\bf Input:} An integer multiset $X = \{x_{1},\, x_{2},\, \dots,\, x_{n}\}$ and an integer target $t$.
    
    \vspace{.05in}
    {\bf Setup:} Constants $\epsilon_1 \approx 0.1579$ and $\epsilon_2 \approx 0.2427$, the solutions to \Cref{eq:eps2,eq:eps3}. \\
    \white{\bf Setup:} Parameters $\beta \approx 1.1186$, $\lambda \approx 0.0202$, $s(\word)$, and $k(n,\, \word)$, all to be specified in the proof. \\
    \white{\bf Setup:} A uniform random prime modulus $\bp \sim \P[\word^{1 + \beta / 2}: 2\word^{1 + \beta / 2}]$.
    \begin{enumerate}
    \setcounter{enumi}{-1}
        \item\label{alg:memo-ov:0} 
        Fix any partition of $X = A \cup B \cup C$ such that $|C| = \beta\log(\frac{\word}{\beta\log\word})$ and $|A| = |B| = \frac{n - |C|}{2}$.
        
        \item\label{alg:memo-ov:1}
        Use \Cref{lemma:input2} to solve $(X, t)$ if there exists a solution $S \subseteq X$ with $|S \cap C| \notin (1 \pm \epsilon_1)\frac{|C|}{2}$.
        
        \item\label{alg:memo-ov:2}
        Use \Cref{lemma:input1} to solve $(X, t)$ if $|W(C)| \leq 2^{|C|} \cdot \word^{-\lambda}$.
        
        \item\label{alg:memo-ov:3}
        Create the sorted lists $L_A$ and $L_B$ using \blackref{alg:SSE}.
        Let $\{\bL_{A,\, i}\}_{i \in [\bp]}$ be the sorted sublists given by $\bL_{A,\, i} := \{a \in L_{A} \mid a \equiv_{\bp} i\}$ and likewise for $\{\bL_{B,\, i}\}_{i \in [\bp]}$. \\
        Create the collection $\calQ^{+\epsilon_2}(C) = \{ Q \mid Q \subseteq C \text{ and } |Q| \leq (1 + \epsilon_2)\frac{|C|}{4}\}$.
        
        \item\label{alg:memo-ov:5}
        Repeat Lines~\ref{alg:memo-ov:6} to \ref{alg:memo-ov:7} either $s(\word)$ times or until a total of $k(n,\, \word)$ many sum-subset couples have been created in Line~\ref{alg:memo-ov:6} (whichever comes first):
        
        \item\label{alg:memo-ov:6}
        \qquad Sample a uniform random residue $\boldr \sim [\bp]$. Then use \Cref{lemma:sse-quarter} and the subroutine \\
        \qquad \blackref{alg:residue-couple-list} to create the sorted lists $\bR_{A,\, \boldr}$ and $\bR_{B,\, \boldr}$: \\
        \qquad $\bR_{A,\, \boldr} = \{(a' := a + \Sigma(Q_1),\, Q_1) \mid (a,\, Q_1) \in L_{A} \times \calQ^{+\epsilon_2}(C) ~\text{with}~ a' \equiv_{\bp} \boldr \}$ and \\
        \qquad $\bR_{B,\, \boldr} = \{(b' := b + \Sigma(Q_2),\, Q_2) \mid (b,\, Q_2) \in L_{B} \times \calQ^{+\epsilon_2}(C) ~\text{with}~ b' \equiv_{\bp} (t - \boldr) \}$. \\
        \algcomment{In fact, $\bR_{A,\, \boldr}$ and $\bR_{B,\, \boldr}$ are stored in a succinct format by \colorref{OliveGreen}{alg:residue-couple-list}.}
        
        \item\label{alg:memo-ov:7}
        \qquad Use \Cref{lemma:sse-quarter} (based on \blackref{alg:MeetInTheMiddle} and the Boolean function \blackref{function:OV}) \\
        \qquad to search the sorted lists $\bR_{A,\, \boldr}$ and $\bR_{B,\, \boldr}$ for a solution $(a',\, Q_1)$, $(b',\, Q_2)$ such that \\
        \qquad $a' + b' = t$ and $Q_1 \cap Q_2 = \emptyset$. Halt and return ``yes'' if a solution is found.
        
        \item\label{alg:memo-ov:8} Return ``no'' (i.e., no solution was found).
        \algcomment{Possibly a false negative.}
    \end{enumerate}
    \end{flushleft}
    \end{mdframed}
    \caption{The {\RepresentationOV} algorithm.
    \label{fig:alg:memo-ov}}
\end{figure}

Now we are ready to state and prove our result about \blackref{alg:memo-ov}:

\begin{theorem}
\label{thm:memo-ov}
\blackref{alg:memo-ov} is a one-sided error randomized algorithm for the Subset Sum problem (with no false positives) with worst-case runtime $O(2^{n/2} \cdot \word^{-\gamma}) \leq O(2^{n/2} \cdot n^{-\gamma})$, for some constant $\gamma > 0.01$, and a 
success probability at least $1/3$ in the circuit RAM model.
\end{theorem}

\begin{proof}[Proof of Correctness for \blackref{alg:memo-ov}]
The constants $\epsilon_1 \approx 0.1579$ and $\epsilon_2 \approx 0.2427$ given in the description of the algorithm are the solutions to the following equations:
\begin{align}
    \tfrac{1 - H((1 - \epsilon_1) / 2)}{(1 - \epsilon_1) / 2}
    & ~=~ 1 - H(\tfrac{1 - \epsilon_2}{2}),
    \label{eq:eps2} \\
    1 + \epsilon_1 - 3H(\tfrac{1 - \epsilon_1}{2})
    & ~=~ -2H(\tfrac{1 + \epsilon_2}{4}).
    \label{eq:eps3}
\end{align}
Also, we set $\beta := \frac{1}{H((1 + \epsilon_2) / 4)} \approx 1.1186$ and $\lambda := (1 - 10^{-5}) \cdot \frac{1 - \epsilon_1}{2} \cdot \beta \cdot (1 - H(\frac{1 - \epsilon_2}{2})) \approx 0.0202$.

\vspace{.1in}
Lines~\ref{alg:memo-ov:1} and \ref{alg:memo-ov:2} preprocess the instance $(X,\, t)$, solving it deterministically\ignore{ in the claimed runtime} via \Cref{lemma:input1,lemma:input2} unless both Conditions~\blackref{con:memo-ov:1} and \blackref{con:memo-ov:2} hold: 

\vspace{.05in}
\noindent
{\bf Condition~\term[(1)]{con:memo-ov:1}:}
$(X,\, t)$ is either a ``yes'' instance with $|S \cap C| \in [(1 - \epsilon_1)\frac{|C|}{2}: \frac{|C|}{2}]$ for each solution $S \subseteq X$,\footnote{\label{footnote:symmetry}Given Line~\ref{alg:memo-ov:1}, a solution $S \subseteq X$ or its complementary $(X \setminus S)$, as a solution to the complementary instance $(X,\, \Sigma(X)-t)$, must have this property.} or a ``no'' instance.

\vspace{.05in}
\noindent
{\bf Condition~\term[(2)]{con:memo-ov:2}:}
$|W(C)| > 2^{|C|} \cdot \word^{-\lambda}$. Note that this implies $|W(T)| > 2^{|T|} \cdot \word^{-\lambda}$ for each $T \subseteq C$.

\vspace{.05in}
\noindent
Lines~\ref{alg:memo-ov:3} to \ref{alg:memo-ov:8} accept only if a solution $t = a' + b' = (a + \Sigma(Q_1)) + (b + \Sigma(Q_2))$ is found in  Line~\ref{alg:memo-ov:7}, for which $a \in L_{A}$, $b \in L_{B}$, and $Q_1,\, Q_2 \in \calQ^{+\epsilon_2}(C)$ are disjoint.
As a consequence, the algorithm never reports false positives.

\vspace{.1in}
It remains to show that Lines~\ref{alg:memo-ov:5} to \ref{alg:memo-ov:7} accept a ``yes'' instance $(X,\, t)$ with probability at least $1/3$ when $(X,\, t)$ satisfies Conditions~\blackref{con:memo-ov:1} and \blackref{con:memo-ov:2}.
Consider a solution $S \subseteq X$ and the following set $W'$ 
containing distinct sums of all ``$\epsilon_2$-balanced'' subsets of $S \cap C$:
\[
    W' ~:=~ \{\, \Sigma(Q) \mid Q \subseteq (S \cap C) ~\text{and}~ |Q| \in (1 \pm \epsilon_2)\tfrac{|S \cap C|}{2} \,\}.
    \hspace{.39cm}
\]
Line~\ref{alg:memo-ov:3} creates the residue sublists $L_{A} = \{\bL_{A,\, i}\}_{i \in [\bp]}$ and $L_{B} = \{\bL_{B,\, i}\}_{i \in [\bp]}$. We say that a residue $i \in [\bp]$ is {\em good} if it satisfies $i - \Sigma(S \cap A) \in (W' \bmod \bp)$, namely there exists a subset $Q_1 \subseteq (S \cap C)$ of size $|Q_1| \in (1 \pm \epsilon_2)\frac{|S \cap C|}{2}$ such that $\Sigma(S \cap A) + \Sigma(Q_1) \equiv_{\bp} i$.
On sampling a good residue $\boldr = i$ in Line~\ref{alg:memo-ov:6}, both $Q_1$ and $Q_2 := (S \cap C) \setminus Q_1$ are of size at most $(1 + \epsilon_2)\frac{|S \cap C|}{2} \leq (1 + \epsilon_2)\frac{|C|}{4}$, so they are included in the collection $\calQ^{+\epsilon_2}(C)$ and, respectively, in the lists $\bR_{A,\, \boldr}$ and $\bR_{B,\, \boldr}$ created in Line~\ref{alg:memo-ov:6}.
Then using \Cref{lemma:sse-quarter}, we are ensured to find the solution $S = (S \cap A) \cup (S \cap B) \cup (Q_1 \cup Q_2)$.

Hence it suffices to {\bf (i)}~lower bound the probability that at least one of the $s(\word)$ samples $\boldr \sim [\bp]$ is good, and {\bf (ii)}~upper bound the probability that these samples generate a total of $k(n,\, \word)$ or more sum-subset couples.

\vspace{.1in}
\noindent
{\bf (i).}
We claim that the size of set $W'$ is at least $\Omega(2^{|S \cap C|} \cdot \word^{-\lambda}) \geq \tOmega(\word^{(1 - \epsilon_1) \cdot \beta / 2 - \lambda})$ and is at most $2^{|S \cap C|} \leq \word^{\beta / 2}$.
The lower bound on the size comes from a combination of two observations.
First, the set $(S \cap C)$ has at least $2^{|S \cap C|} \cdot \word^{-\lambda}$ many distinct subset sums, by Condition~\blackref{con:memo-ov:2}.
Second, the number of subsets $Q \subseteq (S \cap C)$ of size $|Q| \notin (1 \pm \epsilon_2)\frac{|S \cap C|}{2}$ is at most $2 \cdot 2^{H(\frac{1 - \epsilon_2}{2}) \cdot |S \cap C|} = o(2^{|S \cap C|} \cdot \word^{-\lambda})$, given Stirling's approximation (\Cref{eq:stirling}) and the technical condition
$$\textstyle\frac{|S \cap C|}{\log \word} \cdot (1 - H(\frac{1 - \epsilon_2}{2}))
~\geq~ (1 - o_{n}(1)) \cdot \frac{1 - \epsilon_1}{2} \cdot \beta \cdot (1 - H(\frac{1 - \epsilon_2}{2}))
~>~ \lambda,$$
which is true for our choice of $\lambda$.

The upper and lower bounds on $|W'|$ allow us to apply \Cref{lemma:prime-dist}:   
with probability at least $3 / 4$ over the modulus $\bp \sim \P[\word^{1 + \beta / 2}: 2\word^{1 + \beta / 2}]$, there are $|W' \bmod \bp| = \Omega(|W'|) = \tOmega(\word^{(1 - \epsilon_1) \cdot \beta / 2 - \lambda})$ many good residues.
Conditioned on this event, taking 
\[
    s(\word) ~:=~ \tTheta(\word^{1 + \lambda + \epsilon_1 \cdot \beta / 2})
    ~\geq~ \tOmega\left(\tfrac{\bp}{|(W' \bmod \bp)|}\right) 
\]
many samples $\boldr \sim [\bp]$ yields at least one good residue with probability $\geq 2 / 3$. 
\medskip

\noindent
{\bf (ii).}
All candidate lists $\{\bR_{A,\, i}\}_{i \in [\bp]}$, $\{\bR_{B,\, i}\}_{i \in [\bp]}$ together have a total of $$(|L_{A}| + |L_{B}|) \cdot |\calQ^{+\epsilon_2}(C)| \leq 2 \cdot 2^{(n - |C|)/2} \cdot 2^{|C| / \beta} = O(2^{n/2} \cdot \word^{1-\beta/2})$$ sum-subset couples, by construction (Line~\ref{alg:memo-ov:3}), the choices of $|A|$, $|B|$, $|C|$, and Stirling's approximation (\Cref{eq:stirling}).
For any outcome of the random modulus $\bp = p = \Theta(\word^{1+\beta/2})$, a number of $s(\word)$ samples $\boldr \sim [p]$ generate a total of $O(2^{n/2} \cdot \word^{1-\beta/2}) \cdot s(\word) / p = \tO(2^{n/2} \cdot \word^{-((1-\epsilon_1 / 2) \cdot \beta - (1+ \lambda))})$ sum-subset couples in expectation. By setting a large enough cutoff
\[
    k(n,\, \word) ~:=~ \tTheta(2^{n/2} \cdot \word^{-((1-\epsilon_1 / 2) \cdot \beta - (1+ \lambda))}),
    \hspace{2.8cm}
\]
the probability that $k(n,\, \word)$ or more sum-subset couples are generated is at most $1/6$.

Overall, our algorithm succeeds with probability $\geq (3/4) \cdot (2/3) - (1/6) = 1/3$.
\end{proof}

\begin{proof}[Proof of Runtime for \blackref{alg:memo-ov}]
Recalling the assumption $\word = \poly(n)$:
\begin{flushleft}
\begin{itemize}
    \item Line~\ref{alg:memo-ov:1} takes time $\tO(2^{n/2} \cdot \word^{-(1 - H((1 - \epsilon_1) / 2)) \cdot \beta / 2})$, by \Cref{lemma:input2}.
    
    \item Line~\ref{alg:memo-ov:2} takes time $\tO(2^{n/2} \cdot \word^{-\lambda/2})$, by \Cref{lemma:input1}.
    
    \item Line~\ref{alg:memo-ov:3} takes time $O(2^{|A|} + 2^{|B|} + 2^{|C|}) \cdot O(\log\word) = \tO(2^{n/2} \cdot \word^{-\beta / 2})$, by \Cref{lem:sse}, the choices of $|A|$, $|B|$, $|C|$, and that the modulo operation in the circuit RAM model takes time $O(\log \word)$.
    
    \item Lines~\ref{alg:memo-ov:5} to \ref{alg:memo-ov:7} take time $\tO(k(n,\, \word)) = \tO(2^{n/2} \cdot \word^{-((1-\epsilon_1 / 2) \cdot \beta - (1+ \lambda))})$, because a single iteration (\Cref{lemma:sse-quarter}) takes time $\tO(|\bR_{A,\, \boldr}| + |\bR_{B,\, \boldr}|)$ and by construction we create a total of at most $\sum_{\boldr} (|\bR_{A,\, \boldr}| + |\bR_{B,\, \boldr}|) \leq k(n,\, \word)$ many sum-subset couples.
\end{itemize}
\end{flushleft}
The runtime of Line~\ref{alg:memo-ov:3} is dominated by that of Line~\ref{alg:memo-ov:1}, so the bottleneck occurs in Line~\ref{alg:memo-ov:1}, Line~\ref{alg:memo-ov:2}, or Lines~\ref{alg:memo-ov:5} to \ref{alg:memo-ov:7}.
For the choices of constants given in the algorithm, we achieve a speedup of $\Omega(\word^{\gamma})$ for any constant $\gamma \in (0,\, \gamma_{*})$, where the
\[
    \textstyle
    \gamma_{*}
    ~:=~ \min\Big\{\,
    \lambda / 2,\quad
    (1 - H(\frac{1 - \epsilon_1}{2})) \cdot \beta / 2,\quad
    (1 - \epsilon_1 / 2) \cdot \beta - (1+ \lambda)\,\Big\}
    ~=~ \lambda / 2 ~\approx~ 0.0101.
    \qedhere
\]
\end{proof}

\subsection{Auxiliary Lemmas}
\label{subsec:auxiliary}

A first ``preprocessing lemma'' stems from the observation that a not-too-large
subset $Y \subseteq X$ with few distinct subset sums (i.e., $|W(Y)| \ll 2^{|Y|}$) can help  speed up \blackref{alg:MeetInTheMiddle}. This idea goes back to \cite{austrin2015subset,austrin2016dense}, although the condition of ``few distinct subset sums'' we use refers to sets with polynomially rather than exponentially fewer subset sums than the maximum possible number.

\begin{lemma}[Speedup via Additive Structure]
\label{lemma:input1}
Let $(X,\, t)$ be an $n$-integer Subset Sum instance. Given as input $(X,\, t)$ and a subset $Y \subseteq X$ of size $|Y| \leq n / 2$ such that $|W(Y)| \leq 2^{|Y|} \cdot \word^{-\epsilon}$ for some constant $\epsilon > 0$, the instance $(X,\, t)$ can be solved deterministically in time $\tO(2^{n/2} \cdot \word^{-\epsilon / 2})$.
\end{lemma}

\begin{proof}
Fix any size-$(\frac{n + \epsilon \log \word}{2})$ subset $A$ such that $Y \subseteq A \subseteq X$; we have $|W(A)| \leq 2^{|A \setminus Y|} \cdot |W(Y)| \leq 2^{n / 2} \cdot \word^{-\epsilon / 2}$. We also have $|W(X \setminus A)| \leq 2^{|X \setminus A|} = 2^{n/ 2} \cdot \word^{-\epsilon / 2}$. By \Cref{lem:sse,lemma:sse-poly}, it takes overall time $O(2^{n/2} \cdot \word^{-\epsilon / 2} \cdot \log n)$ in the regime $\word = \poly(n)$ to create the sorted lists $L_{A}$, $L_{X \setminus A}$ and to run \blackref{alg:MeetInTheMiddle}.
\end{proof}

Another useful preprocessing lemma shows that the existence of a solution that is ``unbalanced'' vis-a-vis a given small subset yields a speedup:

\begin{lemma}[Speedup via Unbalanced Solutions]
\label{lemma:input2}
\begin{flushleft}
Let $(X,\, t)$ be an $n$-integer Subset Sum instance that has a solution. Given as input $(X,\, t)$ and a subset $Y \subseteq X$ of size $|Y| = \beta \log(\word / (\beta\log\word))$ for $\beta$ as specified in \blackref{alg:memo-ov} such that some solution $S \subseteq X$ satisfies $|S \cap Y| \notin (1 \pm \epsilon)\frac{|Y|}{2}$ for some constant $\epsilon > 0$, the solution $S$ can be found deterministically in time $\tO(2^{n/2} \cdot \word^{-\delta / 2})$, where the constant $\delta := (1 - H(\frac{1 - \epsilon}{2})) \cdot \beta$.
\end{flushleft}
\end{lemma}

\begin{proof}
We can assume without loss of generality that the solution $S$ satisfies $|S \cap Y| \leq \frac{|Y|}{2}$ (since either the original instance $(X,\,t)$ or the complementary instance $(X,\, \Sigma(X)-t)$ must satisfy this property, and we can attempt both instances and only double the runtime).
Hence we can suppose $|S \cap Y| \leq (1 - \epsilon)\frac{|Y|}{2}$. Then the sorted list $L'_Y$ of the set $\{\Sigma(T) \mid T \subseteq Y ~\text{and}~ |T| \leq (1 - \epsilon)\frac{|Y|}{2}\}$ can be created in time $O(2^{|Y|}) = O(\word^{\beta}) = \poly(n)$
using \blackref{alg:SSE} (restricted to subsets of sizes $< (1 - \epsilon)\frac{|Y|}{2}$).

Fix any size-$(\frac{n + (\delta/\beta)|Y|}{2})$ subset $A$ with $Y \subseteq A \subseteq X$. In the regime $\word = \poly(n)$, we can use $L'_{Y}$ and \Cref{lemma:sse-poly} to create the sorted list $L'_A$ of  $\{\Sigma(T) \mid T \subseteq A ~\text{and}~ |T \cap Y| < (1 - \epsilon)\frac{|Y|}{2}\}$ in time $$\tO(|L'_{A}|)
\leq \tO(2^{|A \setminus Y|} \cdot |L'_Y|)
\leq \tO(2^{|A \setminus Y|} \cdot 2^{H(\frac{1 - \epsilon}{2}) \cdot |Y|})
\leq \tO(2^{n / 2} \cdot \word^{-\delta / 2}),$$ following Stirling's approximation (\cref{eq:stirling}) and the choices of $|Y|$ and $|A|$. Moreover, we can use \Cref{lem:sse} to create the sorted list $L_{X \setminus A}$ in time $O(2^{|X \setminus A|}) = \tO(2^{n / 2} \cdot \word^{-\delta / 2})$.


Provided $|S \cap Y| < (1 - \epsilon)\frac{|Y|}{2}$, running \blackref{alg:MeetInTheMiddle} on $L'_{A}$ and $L_{X \setminus A}$ solves $(X, t)$ and takes time $O(|L'_A| + |L_{X \setminus A}|) = \tO(2^{n/2} \cdot \word^{-\delta / 2})$. The overall runtime is $\tO(2^{n/2} \cdot \word^{-\delta / 2})$.
\end{proof}

The next lemma specifies a parameter space within which any set of distinct integers is likely to fall into many residue classes modulo a random prime. This allows us to reduce the search space by considering only solutions that fall into certain residue classes.

\begin{lemma}[Distribution of Integer Sets modulo Random Primes]
\label{lemma:prime-dist}
Fix a set $Y$ of at most $|Y| \leq \word^{\beta / 2}$ distinct $\word$-bit integers for $\beta$ as specified in \blackref{alg:memo-ov}. For a uniform random modulus $\bp \sim \P[\word^{1 + \beta / 2}: 2\word^{1 + \beta / 2}]$, the residue set has size $|(Y \bmod \bp)| = \Theta(|Y|)$ with probability at least $3 / 4$.
\end{lemma}

\begin{proof}
By the prime number theorem \cite[Equation~(22.19.3)]{hardy1979introduction}, there are at least $\frac{\word^{1 + \beta / 2}}{(1 + \beta / 2) \cdot \log \word}$ primes in $\P[\word^{1 + \beta / 2}: 2\word^{1 + \beta / 2}]$ (for any sufficiently large $\word$). Given any two distinct integers $y \neq z \in Y$, the difference $|y - z| \leq 2^{\word}$ has at most $\frac{\word}{(1 + \beta / 2) \cdot \log \word}$ distinct prime factors in $\P[\word^{1 + \beta / 2}: 2\word^{1 + \beta / 2}]$.
Thus under a uniform random choice of modulus $\bp \sim \P[\word^{1 + \beta / 2}: 2\word^{1 + \beta / 2}]$, the second frequency moment $\boldf_{(2)} := |\{(y,\, z) \mid y,\, z \in Y ~\text{with}~ y \equiv_{\bp} z\}|$ has the expectation
\begin{align*}
    \E_{\bp} [\boldf_{(2)}]
    ~=~ |\{y = z \in Y\}| + |\{y \neq z \in Y\}| \cdot \word^{-\beta / 2}
    ~=~ |Y| + (|Y|^{2} - |Y|) \cdot \word^{-\beta / 2}
    ~\leq~ 2|Y|.
\end{align*}
Therefore, with an arbitrarily high constant probability, we have $\boldf_{(2)} = O(|Y|)$ and, by the Cauchy-Schwarz inequality $|Y \bmod \bp| \cdot \boldf_{(2)} \geq |Y|^{2}$,
a residue set of size $|Y \bmod \bp| = \Omega(|Y|)$.
\end{proof}

\begin{figure}[t]
    \centering
    \begin{mdframed}
    Subroutine $\term[\ResidueCoupleList]{alg:residue-couple-list}(\{L_{A,\, i}\}_{i \in [p]},\, \calQ^{+\epsilon}(C),\, r)$

    \begin{flushleft}
    {\bf Input:}
    A collection of $p = \poly(\word)$ sorted sublists $L_{A} = \bigcup_{i \in [p]} L_{A,\, i}$, for some subset $A \subseteq X$, indexed by residue class modulo $p$. \\
    
    \vspace{.05in}
    {\bf Output:}
    A sorted sum-subset list $R_{A,\, r}$ with elements in the sum-collection format $(a',\, \calQ_{a'})$.
    
    \begin{enumerate}[label = 5(\alph*)]
        \item\label{alg:residue-couple-list:2}
        For each $Q \in \calQ^{+\epsilon}(C)$, create the sum-subset sublist $R_{Q} := f_{Q}(L_{A,\, j(Q)})$ by applying $f_{Q}$, the element-to-couple operation $a \mapsto (a' := a + \Sigma(Q),\, Q)$, to the particular input sublist of index $j(Q) := ((r - \Sigma(Q)) \bmod p)$.
        \algcomment{Each $R_{Q}$ is sorted by the sums $a'$.}
        
        \item\label{alg:residue-couple-list:3}
        Let $R_{A,\, r}$ be the list obtained by merging $\{R_{Q}\}_{Q \in \calQ^{+\epsilon}(C)}$, sorted by sums $a' = a + \Sigma(Q)$. \\
        For each distinct sum $a'$, compress all couples $(a',\, Q_1)$, $(a',\, Q_2)$, $\dots$ with the same first element $= a'$ into a single data object $(a',\, \calQ_{a'} := \{Q_1, Q_2, \dots\})$. \\
        \algcomment{Note that for each sum $a'$ we have $a' \equiv_{p} r$ and $|\calQ_{a'}| \leq |\calQ^{+\epsilon}(C)|$.}
    \end{enumerate}
    \end{flushleft}
    \end{mdframed}
    \caption{The {\ResidueCoupleList} subroutine.
    \label{fig:alg:residue-couple-list}}
\end{figure}

\begin{lemma}[Sorted Lists $\bR_{A,\, \boldr}$ and $\bR_{B,\, \boldr}$; {Lines~\ref{alg:memo-ov:6} and~\ref{alg:memo-ov:7}}]
\label{lemma:sse-quarter}
In the context of \blackref{alg:memo-ov}:
{\bf (i)}~Line~\ref{alg:memo-ov:6} uses the subroutine \blackref{alg:residue-couple-list} to create the sorted lists $\bR_{A,\, \boldr}$, $\bR_{B,\, \boldr}$ in time $\tO(|\bR_{A,\, \boldr}| + |\bR_{B,\, \boldr}|)$. Moreover, {\bf (ii)}~Line~\ref{alg:memo-ov:7} finds a solution pair $(a',\, Q_1) \in \bR_{A,\, \boldr}$, $(b',\, Q_2) \in \bR_{B,\, \boldr}$, if one exists, in time $O(|\bR_{A,\, \boldr}| + |\bR_{B,\, \boldr}|)$.
\end{lemma}

\begin{proof}
Line~\ref{alg:memo-ov:6} creates $\bR_{A,\, \boldr}$ and $\bR_{B,\, \boldr}$ using the subroutine \blackref{alg:residue-couple-list} (see \Cref{fig:alg:residue-couple-list}). First, we claim that \blackref{alg:residue-couple-list} takes time $\tO(|\bR_A|)$ in the regime $\word = \poly(n)$:
\begin{flushleft}
\begin{itemize}
    \item Line~\ref{alg:residue-couple-list:2} takes time $\Sigma_{Q \,\in\, \calQ^{+\epsilon_2}(C)} O(|\bR_Q|) = O(|\bR_{A,\, \boldr}|)$, since all shifts $\Sigma(Q)$ for $Q \in \calQ^{+\epsilon_2}(C)$ can be precomputed and memoized when the collection $\calQ^{+\epsilon_2}(C)$ is created in Line~\ref{alg:memo-ov:3}.
    
    \item Line~\ref{alg:residue-couple-list:3} builds one sorted list $\bR_{A,\, \boldr}$ from $|\calQ^{+\epsilon_2}(C)| \leq 2^{|C| / \beta} = \word / (\beta\log\word) = \poly(n)$ sorted sublists $\{\bR_{Q}\}_{Q \,\in\, \calQ^{+\epsilon_2}(C)}$, taking time $O(|\bR_{A,\, \boldr}| \cdot \log n)$ via the classic {\em merge sort} algorithm.
\end{itemize}
\end{flushleft}

After Line~\ref{alg:memo-ov:6} creates the sorted lists $\bR_{A,\, \boldr} = \{(a',\, \bcalQ_{a'})\}$ and $\bR_{B,\, \boldr} = \{(b',\, \bcalQ_{b'})\}$, Line~\ref{alg:memo-ov:7} can run \blackref{alg:MeetInTheMiddle} based on the (ordered) indices $a'$ and $b'$ in time $O(|\bR_{A,\, \boldr}| + |\bR_{B,\, \boldr}|)$.
This ensures that we discover every pair $(a',\, \bcalQ_{a'})$, $(b',\, \bcalQ_{b'})$ such that $a' + b' = t$. Such a pair yields a solution if and only if it contains two disjoint near-quartersets $Q_1 \in \bcalQ_{a'}$, $Q_2 \in \bcalQ_{b'}$ with $Q_1 \cap Q_2 = \emptyset$, which we can check in constant time by one call of the Boolean function \blackref{function:OV}.
Namely, each near-quarterset $Q \in \calQ^{+\epsilon_2}(C)$ is stored in $|C| < \beta \log \word$ bits, so each collection $\bcalQ_{a'},\, \bcalQ_{b'} \subseteq \calQ^{+\epsilon_2}(C)$
can be stored a single word $|C| \cdot |\calQ^{+\epsilon_2}(C)| \leq \word$. Thus one call of \blackref{function:OV} suffices to check a given pair $(a',\, \bcalQ_{a'})$, $(b',\, \bcalQ_{b'})$. Overall, Line~\ref{alg:memo-ov:7} takes time $O(|\bR_{A,\, \boldr}| + |\bR_{B,\, \boldr}|)$.
\end{proof}

\begin{observation}[Adapting \blackref{alg:memo-ov} to Word RAM]
\label{obs:memo-ov-word-RAM}
In the word RAM model, it may take superconstant time to evaluate the Boolean function \blackref{function:OV}. Similar to the strategy used in \Cref{sec:bit-packing}, our word RAM variant avoids this issue by performing \blackref{alg:memo-ov} as if the word length were $\word' := 0.1n$ (using sets $A'$, $B'$, $C'$, etc., with appropriate size modifications), except for three modifications:
\begin{enumerate}
    \item In lines~\ref{alg:memo-ov:3} and \ref{alg:residue-couple-list:3}, on creating a (sub)collection $\calQ \subseteq \calQ^{+\epsilon_2}(C')$ as a bit string, store it in at most $\leq \lceil |\calQ^{+\epsilon_2}(C')| \cdot |C'| / \word \rceil \leq \lceil \word' / \word \rceil = \Theta(1)$ words (since a single word with $\word$ bits may be insufficient).
    
    \item In Line~\ref{alg:memo-ov:3}, after creating the collection $\calQ^{+\epsilon_2}(C')$, create a lookup table \texttt{OV}$'$ that memoizes the input-output result of the Boolean function \blackref{function:OV} on each subcollection pair $\calQ_{a'},\, \calQ_{b'} \subseteq \calQ^{+\epsilon_2}(C')$
    in time $(2^{|\calQ^{+\epsilon_2}(C')|})^2 \cdot \poly(|\calQ^{+\epsilon_2}(C')|) \leq (2^{\word'})^{2} \cdot \poly(\word') = O(2^{0.21n})$. This table can then be accessed using a $2\lceil \word' / \word \rceil = \Theta(1)$-word index in constant time.
    \item Line~\ref{alg:memo-ov:7} replaces the Boolean function \blackref{function:OV} (namely a constant-time $\mathsf{AC}^0$ circuit RAM operation) with constant-time lookup into \texttt{OV}$'$.
\end{enumerate}
Compared with running \blackref{alg:memo-ov} itself for $\word' = 0.1n$, the only difference of this variant is that $\bR_{A,\, \boldr}$ and $\bR_{B,\, \boldr}$ are stored in $\lceil \word' / \word \rceil = \Theta(1)$ times as many words, given $\word = \Omega(n)$. Hence, it is easy to check the correctness and the runtime $O(2^{n/2} \cdot {\word'}^{-\gamma} \cdot \lceil \word' / \word \rceil) = O(2^{n/2} \cdot n^{-\gamma})$, for the same constant $\gamma > 0.01$.
\end{observation}

\newcommand{\bcalC}{\boldsymbol{\calC}}
\newcommand{\SampleList}{\text{\tt Sample-Packing}}
\newcommand{\SearchList}{\text{\tt Sample-Searching}}

\section{ Subset Sum in Time \texorpdfstring{$O(2^{n/2} \cdot n^{-0.5023})$}{} }
\label{sec:main}

The algorithm in this section is a delicate combination of \blackref{alg:BitPacking} and \blackref{alg:memo-ov}. Our main result (see \Cref{thm:main}) demonstrates that problem-specific features of Subset Sum can be exploited to obtain an (additional) nontrivial time savings beyond what can be achieved only by augmenting \blackref{alg:MeetInTheMiddle} with generic bit-packing tricks. Although the bookkeeping to analyze the algorithm of this section is somewhat intricate, many of the ideas in this section were previewed in \Cref{sec:memo-ov}.

We begin by describing a difficulty that arises in the attempt to combine the two building-block algorithms: (i)~\blackref{alg:BitPacking} saves time by removing a subset $D \subseteq X$ from the input, then running a \blackref{alg:MeetInTheMiddle} variant on the resulting subinstance multiple times.
(ii)~\blackref{alg:memo-ov} runs a \blackref{alg:MeetInTheMiddle} variant on multiple subinstances indexed by residue classes modulo a random prime $\bp$.
This presents a problem: to get the time savings from bit packing, we would like to reuse subinstances multiple times, but to get the time savings from the representation method we need to build separate subinstances with respect to each residue class $\pmod{\bp}$ that contains elements of $W(D)$. To solve this problem, we construct $D$ in a way that ensures the elements of $W(D)$ fall into few residue classes $\pmod{\bp}$.
Specifically, we fix $\bp$ and consider two cases.
In one case\ignore{Case~I}, the elements of $X$ fall into few residue classes $\pmod{\bp}$ and it is possible to choose a small set $D$ such that the elements of $D$ (and $W(D)$) fall into very few residue classes $\pmod{\bp}$. In this case we need to construct just $|W(D) \bmod \bp| = \polylog(n)$ distinct subinstances.
In the other case\ignore{Case~II}, the elements of $X$ fall into many residue classes $\pmod{\bp}$. In this case, a carefully selected $D$ satisfies the weaker bound $|W(D) \bmod \bp| = \tO(n^\delta)$ for a small constant $\delta > 0$, increasing the number of subinstances we need to construct. However, the fact that the elements of $X$ fall into many residue classes $\pmod{\bp}$ lets us select a larger set $C$ such that the subset sums in $W(S \cap C)$ distribute well $\pmod{\bp}$ for any solution $S$, which offsets the increase in runtime.

Like the $\mathsf{AC}^0$ operation \blackref{function:OV} in the \blackref{alg:memo-ov} algorithm,
the core of our new algorithm is another $\mathsf{AC}^0$ operation, \blackref{function:packed-OV}, that solves Orthogonal Vectors on small instances. In comparison, \blackref{function:packed-OV} also takes as input two $\word$-bit words containing multiple bit vectors of length $O(\log n)$, but now the bit vectors in either word may come from two or more lists, and each list is indexed by the $m$-bit hash $\bh_m(s)$ of a corresponding subset sum $s$, for $m = 3\log\word$. Hence, \blackref{function:packed-OV} may solve multiple small Orthogonal Vectors instances at once.

\afterpage{
\begin{figure}[t]
\centering
\begin{mdframed}
Procedure $\term[\PackedRepresentationOV]{alg:packed-rep-ov}(X,\, t)$

\begin{flushleft}
{\bf Input:} A multiset $X = \{x_{1},\, x_{2},\, \dots,\, x_{n}\}$ and an integer target $t$.

\vspace{.05in}
{\bf Setup:} Constants $\epsilon_1 \approx 0.1579$ and $\epsilon_2 \approx 0.2427$.
\\
\white{\bf Setup:} Constants $\beta = \beta(\rho)$, $\epsilon'_1 = \epsilon'_1(\rho)$, $\lambda = \lambda(\rho)$, solutions to \Cref{eq:lambda,eq:main}. \\
\white{\bf Setup:} Parameters $e_1 < 0.5$, $e_2$, $s(\word)$, $k(n,\, \word)$ to be specified. \\
\white{\bf Setup:} Modulus $\bp \sim \P[\word^{1+\beta/2}: 2\word^{1+\beta/2}]$. \\
\white{\bf Setup:} A random pseudolinear hash function $\bh_m$ satisfying \Cref{lem:hash}.
\begin{enumerate}
\setcounter{enumi}{-1}
\setlength\itemsep{0.3em}
    \item\label{alg:packed-rep-ov:0}
    Create random disjoint subsets $\bC,\, \bD \subseteq X$ (whose existence is guaranteed) as  \\
    follows, and then fix any partition of $X \setminus (\bC \cup \bD) = \bA \cup \bB$ such that $|\bA| = |\bB| = \frac{n - |\bC| - |\bD|}{2}$.
    
    {\bf \term[{\bf Case~(I)}]{pack-rep-case1}: $|X \bmod \bp| > \word^{\rho} / \log \word$.\quad}
    Let $e_1 := \epsilon'_1$ and $e_2 := 0$. \\
    Use \Cref{lemma:select-C} to find a particular size-$(\frac{1}{2}\log(\word^{\rho} / \log \word))$ subset $\bC \subseteq X$. \\
    Use \Cref{lemma:select-D} to find a particular size-$((2 - \rho + \beta/2)\log \word)$ subset $\bD \subseteq (X \setminus \bC)$.
    
    {\bf \term[{\bf Case~(II)}]{pack-rep-case2}: $|X \bmod \bp| \leq \word^{\rho} / \log \word$.\quad}
    Let $e_1 := \epsilon_1$ and $e_2 := \epsilon_2$. \\
    Find a size-$(\log \word)$ subset $\bD \subseteq X$ of congruent integers $|\bD \bmod \bp| = 1$. \\
    Sample a uniform random subset $\bC \sim \{T \mid T \subseteq (X \setminus \bD) ~\text{and}~ |T| = \beta \log(\word / (\beta \log \word))\}$.
    
    \item\label{alg:packed-rep-ov:1}
    Use \cref{lemma:input2-2} to solve $(X,\, t)$ if there exists a solution $S \subseteq X$ with $|S \cap \bC| \notin (1 \pm e_1)\frac{|\bC|}{2}$.
    
    \item\label{alg:packed-rep-ov:2}
    In \blackref{pack-rep-case2} only, use \Cref{lemma:input1-2} to solve $(X,\, t)$ if $|W(\bC)| \leq 2^{|\bC|} \cdot \word^{-\lambda}$.
    
    \item\label{alg:packed-rep-ov:3}
    Create the set $\bW(\bD)$ and the collection $\bcalQ^{+e_2}(\bC) = \{ Q \mid Q \subseteq \bC \text{ and } |Q| \leq (1 + e_2)\frac{|\bC|}{4}\}$. \\
    Create the sorted lists $\bL_{\bA}$, $\bL_{\bB}$ using \blackref{alg:SSE}.
    Let $\{\bL_{\bA,\, i}\}_{i \in [\bp]}$ be the sorted sublists given by $\bL_{\bA,\, i} := \{a \in \bL_{\bA} \mid a \equiv_{\bp} i\}$ and likewise for $\{\bL_{\bB,\, i}\}_{i \in [\bp]}$.
    
    \item\label{alg:packed-rep-ov:4}
    For each $t' \in ((t - \bW(\bD)) \bmod \bp)$:
    
    \item\label{alg:packed-rep-ov:5}
    \qquad Repeat Line~\ref{alg:packed-rep-ov:6} either $s(\word)$ times or until a total of $k(n,\, \word)$ many sum-subset \\
    \qquad couples have been created (whichever comes first):
    
    \item\label{alg:packed-rep-ov:6}
    \qquad\qquad Sample a uniform random residue $\boldr \sim [\bp]$. Then create the sorted  \\
    \qquad\qquad sum-subset lists,\\
    $\bR_{\bA,\, t',\, \boldr} = \{(a' := a + \Sigma(Q_1), Q_1) \mid (a, Q_1) \in \bL_{\bA} \times \bcalQ^{+e_2}(\bC) ~\text{with}~ a' \equiv_{\bp} \boldr \}$, \\
    $\bR_{\bB,\, t',\, \boldr} = \{(b' := b + \Sigma(Q_2), Q_2) \mid (b, Q_2) \in \bL_{\bB} \times \bcalQ^{+e_2}(\bC) ~\text{with}~ b' \equiv_{\bp} (t' - \boldr) \}$, \\
    \qquad\qquad in the sum-collection format $\bR_{\bA,\, t',\, \boldr} = \{(a',\, \bcalQ_{a'})\}$, $\bR_{\bB,\, t',\, \boldr} = \{(b',\, \bcalQ_{b'})\}$.
    
    \item\label{alg:packed-rep-ov:7}
    \qquad Let $\bR_{\bA,\, t'}$ be the list obtained by merging $\bigcup_{\boldr} \bR_{\bA,\, t',\, \boldr}$, sorted by sums $a'$, and\\  
    \qquad likewise for $\bR_{\bB,\, t'}$. Then use \Cref{lemma:pack-main}\ignore{ and the subroutine \blackref{alg:sample-list}} to create the hashed and packed lists\\
    \qquad  $\bH_{\bA,\, t'}$, $\bH_{\bB,\, t'}$ from $\bR_{\bA,\, t'}$, $\bR_{\bB,\, t'}$, respectively.
    
    \item\label{alg:packed-rep-ov:8}
    For each $t'' \in (t - \bW(\bD))$:
    
    \item\label{alg:packed-rep-ov:9}
    \qquad Use \Cref{lemma:MiM-main} 
    to search the sorted lists $\bH_{\bA,\, (t'' \bmod \bp)}$, $\bH_{\bB,\, (t'' \bmod \bp)}$, \\
    \qquad  $\bR_{\bA,\, (t'' \bmod \bp)}$, $\bR_{\bB,\, (t'' \bmod \bp)}$ 
    for a solution pair $(a',\, Q_1)$, $(b',\, Q_2)$ such that \\ 
    \qquad $a' + b' = t''$ and $Q_1 \cap Q_2 = \emptyset$. Halt and return ``yes'' if a solution is found. \\
    
    \item\label{alg:packed-rep-ov:10}
    Return ``no'' (i.e., no solution was found).
    \algcomment{Possibly a false negative.}
\end{enumerate}
\end{flushleft}
\end{mdframed}
\caption{The {\PackedRepresentationOV} algorithm.  \label{fig:alg:packed-rep-ov}}
\end{figure}
\clearpage}

\begin{theorem}
\label{thm:main}
    \PackedRepresentationOV is a one-sided error randomized algorithm for the Subset Sum problem (with no false positives) with worst-case runtime $O(2^{n/2} \cdot n^{-(1/2 + \gamma)})$, for some constant $\gamma > 0.0023$, and a success probability at least $1/12$ in the circuit RAM model.
\end{theorem}

Define the density constant $\rho := \frac{\log n}{\log \word}$, which is bounded as $0 < \rho \leq 1 + \Theta(\frac{1}{\log n})$ given $\word = \poly(n)$ and $\word = \Omega(n)$.\footnote{Another density parameter $n / \log t$ is widely used in the literature (e.g.,\, \cite{austrin2016dense}), which can be regarded as $n / \word$ in our language (since $\word = \Omega(\log t)$). In comparison, our density constant $\rho = \frac{\log n}{\log \word}$ is more ``fine-grained''.}






In the density range $0 < \rho \leq (2 - H(1/4))^{-1} \approx 0.8412$, \blackref{alg:packed-rep-ov} can achieve the claimed speedup $\tilde{\Omega}(\word^{1/2}) = \tilde{\Omega}(n^{1/(2\rho)}) \geq \Omega(n^{0.5943}) \geq \Omega(n^{0.5023})$ by naively simulating \blackref{alg:BitPacking}, so it suffices to consider the remaining range $(2 - H(1/4))^{-1} < \rho \leq 1 + \Theta(1 / \log n)$.

\begin{proof}[Proof of Correctness for \PackedRepresentationOV]
The constant $\lambda = \lambda(\rho)$ given in the description of the algorithm is defined by the following equation:
\begin{equation}
    \label{eq:lambda}
    \textstyle
    \lambda ~=~ \lambda(\rho) ~:=~ (1 - 10^{-5}) \cdot \frac{1 - \epsilon_1}{2} \cdot (1 - H(\frac{1 - \epsilon_2}{2})) \cdot \beta
    ~\approx~ 1.8065 \times 10^{-2} \cdot \beta(\rho).
\end{equation}
The other two constants $\beta = \beta(\rho)$, $\epsilon'_1 = \epsilon'_1(\rho)$ will be specified later.

\vspace{.1in}
Lines~\ref{alg:packed-rep-ov:1} and \ref{alg:packed-rep-ov:2} preprocess the instance $(X,\, t)$, solving it by either or both of \Cref{lemma:input1-2,lemma:input2-2}, unless the outcome of the random subset $\bC = C$ satisfies Conditions~\blackref{con:packed-rep-ov:1} and \blackref{con:packed-rep-ov:2} in both cases.
\Cref{lemma:input1-2} applies in \blackref{pack-rep-case2}, provided $2^{|C|} = 2^{\beta \log(\word / (\beta \log \word))} = \tOmega(\word^\beta) > \word^{\lambda}$ by \Cref{eq:lambda}.
(In \blackref{pack-rep-case1} we have $|W(C) \pmod \bp| = 2^{|C|}$ by \Cref{lemma:select-C}, which is even stronger than Condition \blackref{con:packed-rep-ov:2}.)

\vspace{.05in}
\noindent
{\bf Condition~\term[(1)]{con:packed-rep-ov:1}:}
$(X,\, t)$ is either a ``yes'' instance with $|S \cap C| \in [(1 - e_1)\frac{|C|}{2}: \frac{|C|}{2}]$ for each solution $S \subseteq X$ (cf.\ \Cref{footnote:symmetry}), or a ``no'' instance.

\vspace{.05in}
\noindent
{\bf Condition~\term[(2)]{con:packed-rep-ov:2}:}
$|W(C)| > 2^{|C|} \cdot \word^{-\lambda}$. Note that this implies $|W(T)| > 2^{|T|} \cdot \word^{-\lambda}$ for each $T \subseteq C$.

\vspace{.05in}
\noindent
Lines~\ref{alg:packed-rep-ov:3} to \ref{alg:packed-rep-ov:10} accept only if a solution $(a',\, Q_1)$, $(b',\, Q_2)$ is found in Line~\ref{alg:packed-rep-ov:9} for which $a' + b' = t''$ and $Q_1 \cap Q_2 = \emptyset$. As a consequence, the algorithm never makes a false positive mistake.

\vspace{.1in}
It remains to show that Lines~\ref{alg:packed-rep-ov:3} to \ref{alg:packed-rep-ov:9} accept a ``yes'' instance $(X,\, t)$ with constant probability. Fix any outcome of the partition $X = \bA \cup \bB \cup \bC \cup \bD = A \cup B \cup C \cup D$. Consider a solution $S \subseteq X$ and the set $W' = W'(S \cap C)$ of {\em distinct} sums of all the ``$e_2$-balanced'' subsets of $S \cap C$:
\[
    W' ~:=~ \{\, \Sigma(Q) \mid Q \subseteq (S \cap C) ~\text{and}~ |Q| \in (1 \pm e_2)\tfrac{|S \cap C|}{2} \,\}.
\]
Line~\ref{alg:packed-rep-ov:3} creates the residue sublists $L_{A} = \{L_{A,\, i}\}_{i \in [\bp]}$ and $L_{B} = \{L_{B,\, i}\}_{i \in [\bp]}$. We say that a residue $i \in [\bp]$ is {\em good} if it satisfies $i - \Sigma(S \cap A) \in (W' \bmod \bp)$, namely there exists a subset $Q_1 \subseteq (S \cap C)$ of size $|Q_1| \in (1 \pm e_2)\frac{|S \cap C|}{2}$ such that $\Sigma(S \cap A) + \Sigma(Q_1) \equiv_{p} i$.

In the particular iteration $t'_* := ((t - \Sigma(S \cap D)) \bmod \bp)$ of Line~\ref{alg:packed-rep-ov:4}: on sampling a good residue $\boldr = i$ in Line~\ref{alg:packed-rep-ov:6}, both sets $Q_1$ and $Q_2 := (S \cap C) \setminus Q_1$ are of size $\leq (1 + e_2)\frac{|S \cap C|}{2} \leq (1 + e_2)\frac{|C|}{4}$, so they are included in $\calQ^{+e_2}(C)$ and, respectively, in the lists $R_{A,\, t'_*,\, \boldr}$ and $R_{B,\, t'_*,\, \boldr}$ created in Line~\ref{alg:packed-rep-ov:6}.
As a consequence, later in the iteration $t''_* := t - \Sigma(S \cap D)$ of Line~\ref{alg:packed-rep-ov:8}, we can find the solution pair $(a',\, Q_1)$, $(b',\, Q_2)$ with constant probability, by \Cref{lemma:MiM-main}, for which $a' + b' = (\Sigma(S \cap A) + \Sigma(Q_1)) + (\Sigma(S \cap B) + \Sigma(Q_2)) = t''_*$ and $Q_1 \cap Q_2 = \emptyset$.

Hence, it suffices to reason about the iteration $t'_*$ of Line~\ref{alg:packed-rep-ov:4}, {\bf (i)}~lower-bounding the probability that at least one of the $s(\word)$ many samples $\boldr \sim [\bp]$ in Line~\ref{alg:packed-rep-ov:6} is good, and {\bf (ii)}~upper-bounding the probability that these samples generate a total of $k(n,\, \word)$ or more sum-subset couples.


\vspace{.1in}
\noindent
\blackref{pack-rep-case1}{\bf :}
The proof below only relies on Condition~\blackref{con:packed-rep-ov:1}, and works for any outcome of the modulus $\bp = p = \Theta(n^{1+\beta/2})$ and the partition $X = \bA \cup \bB \cup \bC \cup \bD = A \cup B \cup C \cup D$.

\vspace{.1in}
\noindent
{\bf (i).}
By \Cref{lemma:select-C}, we have $|W(C) \bmod p| = 2^{|C|}$, which implies that $|W(S \cap C) \bmod p| = 2^{|S \cap C|}$ and $|W' \bmod p| = \binom{|S \cap C|}{|S \cap C| / 2} = \tOmega(2^{|S \cap C|}) \geq \tOmega(2^{(1 - \epsilon'_1) \cdot |C| / 2}) = \tOmega(\word^{(1 - \epsilon'_1) \cdot \rho / 4})$, by Stirling's approximation, Condition~\blackref{con:packed-rep-ov:1}, and the choice of $|C|$, $e_1$, $e_2$.
Hence, a number of $s(\word) = \Omega(\tfrac{p}{|W' \bmod p|})$ samples $\boldr \sim [p]$ yield at least one good residue with probability $\geq 2 / 3$, for any large enough
\begin{align*}
    s(\word) ~=~ s'(\word) ~:=~ \tTheta(\word^{1 + \beta / 2 - (1 - \epsilon'_1) \cdot \rho / 4}).
\end{align*}

\noindent
{\bf (ii).}
All candidate lists $\{R_{A,\, t'_*,\, i}\}_{i \in [p]}$, $\{R_{B,\, t'_*,\, i}\}_{i \in [\bp]}$ contain a total of $(|L_{A}| + |L_{B}|) \cdot |\calQ^{+e_2}(C)| = \tO(2^{n/2} \cdot \word^{-1/2} \cdot \word^{(H(1 / 4) - 1 / 2) \cdot \rho / 2 - (1 - \rho + \beta / 2) / 2})$ many sum-subset couples, by their construction (Line~\ref{alg:packed-rep-ov:3}), the choices of $|A|$, $|B|$, $|C|$, $e_2$, and Stirling's approximation (\Cref{eq:stirling}).
Thus, a number of $s'(\word)$ many samples $\boldr \sim [p]$ in expectation generate at most $(|L_{A}| + |L_{B}|) \cdot |\calQ^{+e_2}(C)| \cdot s'(\word) / p \leq \frac{1}{3}k(n,\, \word)$ sum-subset couples, for any large enough cutoff
\begin{align*}
    k(n,\, \word) ~=~ k'(n,\, \word)
    ~:=~ \tTheta(2^{n/2} \cdot \word^{-1/2} \cdot \word^{-((1 - H(1 / 4) - \epsilon'_1 / 2) \cdot \rho / 2 + (1 - \rho + \beta / 2) / 2)})
\end{align*}
By Markov's inequality, the probability that $k'(n,\, \word)$ or more couples are generated is at most $1 / 3$.

To conclude, \blackref{pack-rep-case1} succeeds with probability $\geq 2 / 3 - 1 / 3 = 1 / 3$.

\vspace{.1in}
\noindent
\blackref{pack-rep-case2}{\bf .}
The proof below relies on both Conditions~\blackref{con:packed-rep-ov:1} and \blackref{con:packed-rep-ov:2}, and works for any outcome of the created subset $\bD = D \subseteq X$.

Consider the collection $\calC'(X)$ of size-$(\beta\log(\frac{\word}{\beta \log \word}))$ subsets $C' \subseteq X$ that satisfy both conditions:
\[
    \calC'(X) ~:=~ \{C' \subseteq X \mid \text{Conditions~\blackref{con:packed-rep-ov:1} and \blackref{con:packed-rep-ov:2} hold, and $|C'| = \beta\log(\tfrac{\word}{\beta \log \word})$}\},
\]
and also consider its restriction to $(X \setminus D)$:
\[
    \calC(X \setminus D) ~:=~ \{C \subseteq (X \setminus D) \mid \text{Conditions~\blackref{con:packed-rep-ov:1} and \blackref{con:packed-rep-ov:2} hold, and $|C| = \beta\log(\tfrac{\word}{\beta \log \word})$}\}.
\]
Obviously, the conditional random subset $\bC$, given \blackref{pack-rep-case2}
and Conditions~\blackref{con:packed-rep-ov:1} and \blackref{con:packed-rep-ov:2}, is distributed as $\bC \sim \calC(X \setminus D)$.

Let $P_{2} \subseteq \P[\word^{1+\beta/2} : 2\word^{1+\beta/2}]$
be the set of primes 
such that for the input instance $X$ the algorithm enters \blackref{pack-rep-case2} (``\blackref{pack-rep-case2} primes''), and $P_{1}$ the set of primes such that for the input instance $X$ the algorithm enters \blackref{pack-rep-case1} (``\blackref{pack-rep-case1} primes'').
Without loss of generality we suppose $|P_2| \geq 3|P_1|$, i.e., \blackref{pack-rep-case2} occurs with probability $\geq 3 / 4$ over the modulus $\bp \sim \P[\word^{1+\beta/2} : 2\word^{1+\beta/2}]$. (Otherwise, \blackref{pack-rep-case1} occurs with probability $\geq 1 / 4$ and already gives a success probability $\geq 1/4 \cdot 1/3 = 1/12$.)
Obviously, the conditional modulus $\bp$, given \blackref{pack-rep-case2}, is distributed as $\bp \sim P_{2}$.

\vspace{.1in}
\noindent
{\bf (i).}
Each subset $C' \in \calC'(X)$ satisfies $\frac{|S \cap C'|}{\log \word} \cdot (1 - H(\frac{1 - \epsilon_2}{2})) \geq (1 - o_{n}(1)) \cdot \beta \cdot \frac{1 - \epsilon_1}{2} \cdot (1 - H(\frac{1 - \epsilon_2}{2})) > \lambda$, by Condition~\blackref{con:packed-rep-ov:1}, \Cref{eq:lambda}, and the choice of $|C'| = |C|$.
This technical condition, together with Condition~\blackref{con:packed-rep-ov:2}, allows us to use the same arguments as in the proof of correctness for \Cref{thm:memo-ov}, showing that $|W'(S \cap C') \bmod \bp| = \Theta(|W'(S \cap C')|) = \tOmega(\word^{(1 - \epsilon_1) \cdot \beta / 2 - \lambda})$ with probability $\geq 3 / 4$ over the modulus $\bp \sim (P_1 \cup P_2) = \P[\word^{1+\beta/2} : 2\word^{1+\beta/2}]$.
On sampling a random subset $\bC' \sim \calC'(X)$, we have an analogous guarantee\ignore{ on the random set $\bW'(S \cap \bC')$}:
\begin{align*}
    \Pr_{(\bp,\, \bC') \,\sim\, (P_1 \cup P_2) \times \calC'(X)} \big[\, |\bW'(S \cap \bC') \bmod \bp| = \tOmega(\word^{(1 - \epsilon_1) \cdot \beta / 2 - \lambda}) \,\big] ~\geq~ 3 / 4.
\end{align*}
Since there are at least three times as many \blackref{pack-rep-case2} primes as  \blackref{pack-rep-case1} primes, we have
\begin{align*}
    \Pr_{(\bp,\, \bC') \,\sim\, P_2 \times \calC'(X)} \big[\, |\bW'(S \cap \bC') \bmod \bp| = \tOmega(\word^{(1 - \epsilon_1) \cdot \beta / 2 - \lambda}) \,\big] ~\geq~ 2 / 3.
\end{align*}

But the (actual) random subset $\bC$ is distributed as $\bC \sim \calC(X \setminus D)$ rather than $\bC' \sim \calC'(X)$, i.e., we shall also consider the effect of the restriction to $(X \setminus D)$. For the choices of $|C'| = |C|$ and $|D|$, this restriction only removes a negligible fraction of subsets:
\begin{align*}
    \textstyle \frac{|\calC'(X) \setminus \calC(X \setminus D)|}{|\calC'(X)|}
    ~\leq~ 1 - \binom{n - |D|}{|C'|} \big/ \binom{n}{|C'|}
    ~\in~ \frac{\beta\log^{2}(\word)}{n} \cdot (1 \pm o_{n}(1))
    ~=~ o_{n}(1).
\end{align*}
Thus, a uniform random $\bC \sim \calC(X \setminus D)$ yields a large residue set\ignore{ $(\bW'(S \cap \bC) \bmod \bp)$} with constant probability:
\begin{align*}
    \Pr_{(\bp,\, \bC) \,\sim\, P_2 \times \calC(X \setminus D)} \big[\, |\bW'(S \cap \bC) \bmod \bp| = \tOmega(\word^{(1 - \epsilon_1) \cdot \beta / 2 - \lambda}) \,\big] ~\geq~ 2 / 3 - o_{n}(1).
\end{align*}
Conditioned on this event, a number of $s(\word) = \Omega(\tfrac{\bp}{|\bW'(S \cap \bC) \bmod \bp|})$ samples $\boldr \sim [\bp]$ yields at least one good residue with probability $\geq 4 / 5$, for any large enough
\begin{align*}
    s(\word) ~=~ s''(\word)
    ~:=~ \tTheta(\word^{1 + \epsilon_1 \cdot \beta / 2 + \lambda}).
\end{align*}

\noindent
{\bf (ii).}
By adapting the proof of \blackref{pack-rep-case1}, {\bf (ii)} for the current choices of $|A|$, $|B|$, $|C|$, the probability that $k(n,\, \word)$ or more sum-subset couples are generated is at most $1 / 6$, for any large enough cutoff
\[
    k(n,\, \word) ~=~ k''(n,\, \word) ~:=~ \tTheta(2^{n/2} \cdot \word^{-1/2} \cdot \word^{-((1 - H(\frac{1 + \epsilon_2}{4}) - \epsilon_1 / 2) \cdot \beta - \lambda)}).
\]

To conclude, \blackref{pack-rep-case2} succeeds with probability $\geq (2 / 3 - o_{n}(1)) \cdot 4 / 5 - 1 / 6 \geq 1 / 3$.
\end{proof}

\afterpage{
\begin{figure}
    \centering
    \subfloat[\label{fig:main:alpha}
    $\alpha_*(\rho)$]{{\includegraphics[width = .48\textwidth]
    {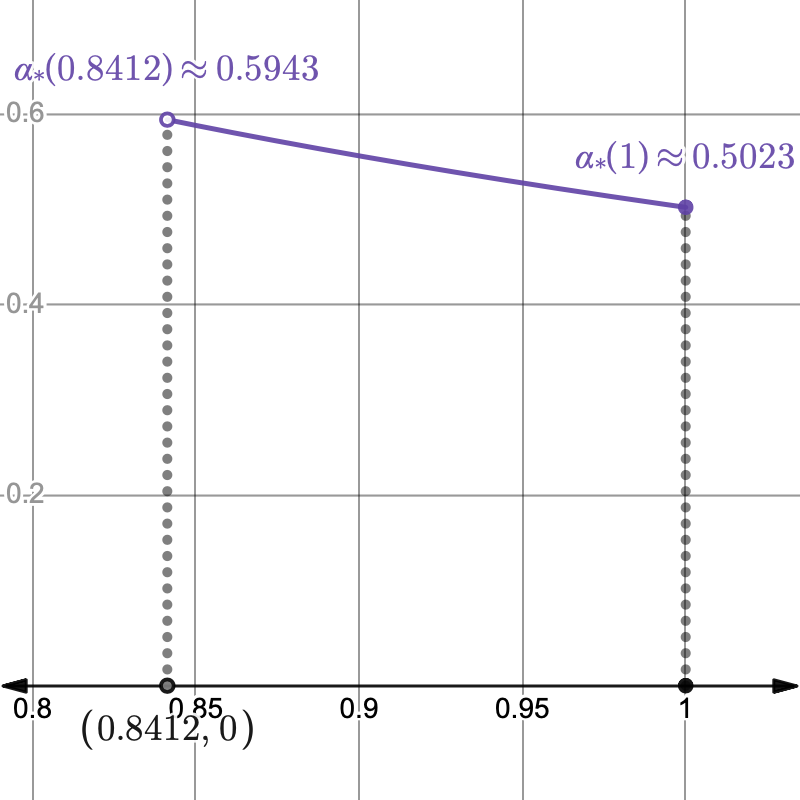}}}
    \hfill
    \subfloat[\label{fig:main:gamma}
    $\gamma_*(\rho)$]{{\includegraphics[width = .48\textwidth]
    {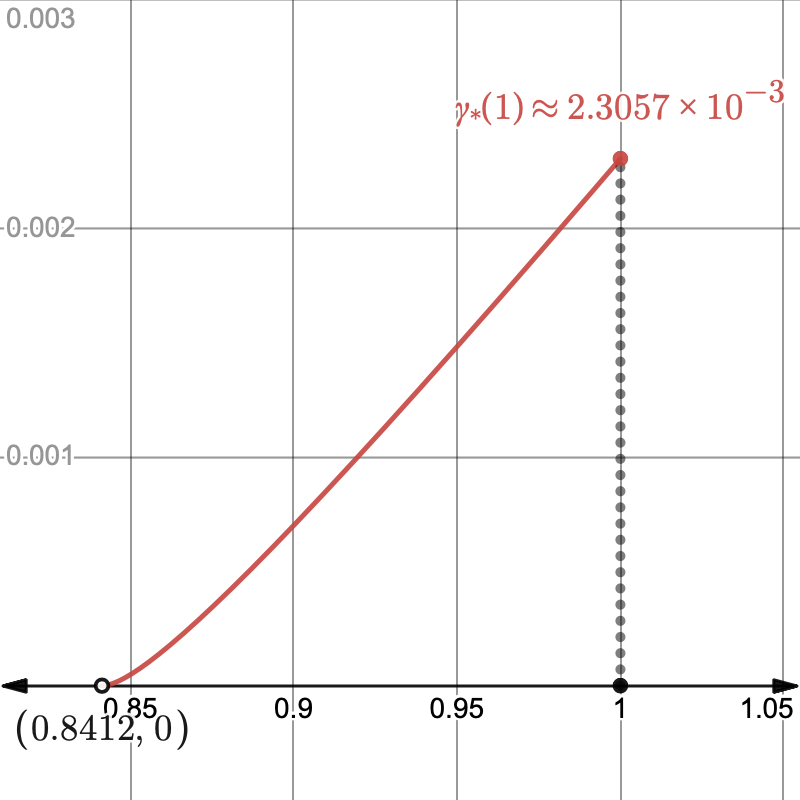}}}
    \\
    \subfloat[\label{fig:main:beta}
    $\beta(\rho)$]{{\includegraphics[width = .48\textwidth]
    {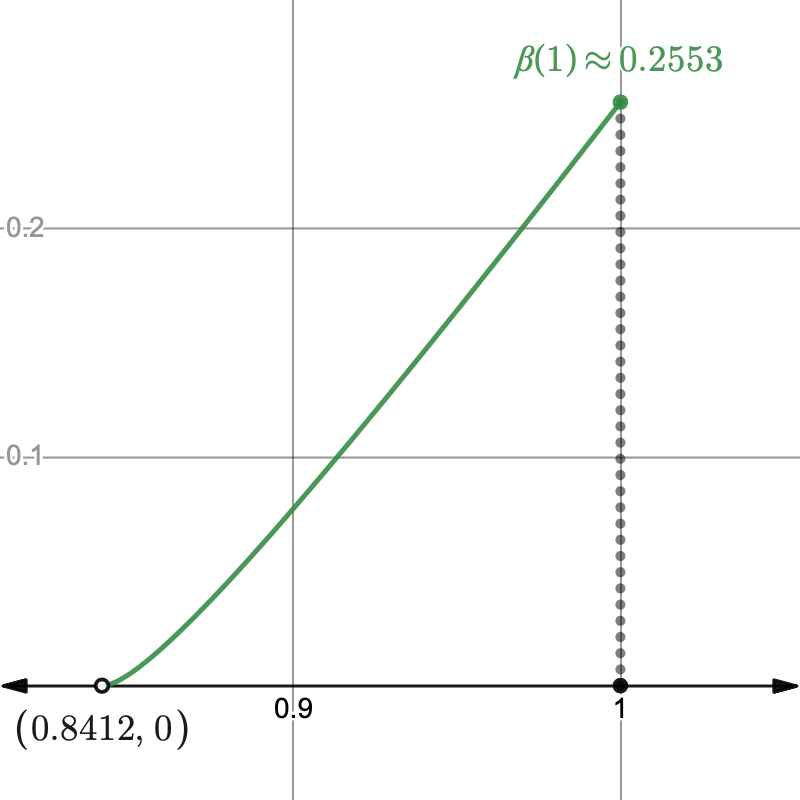}}}
    \hfill
    \subfloat[\label{fig:main:epsilon}
    $\epsilon'_1(\rho)$]{{\includegraphics[width = .48\textwidth]
    {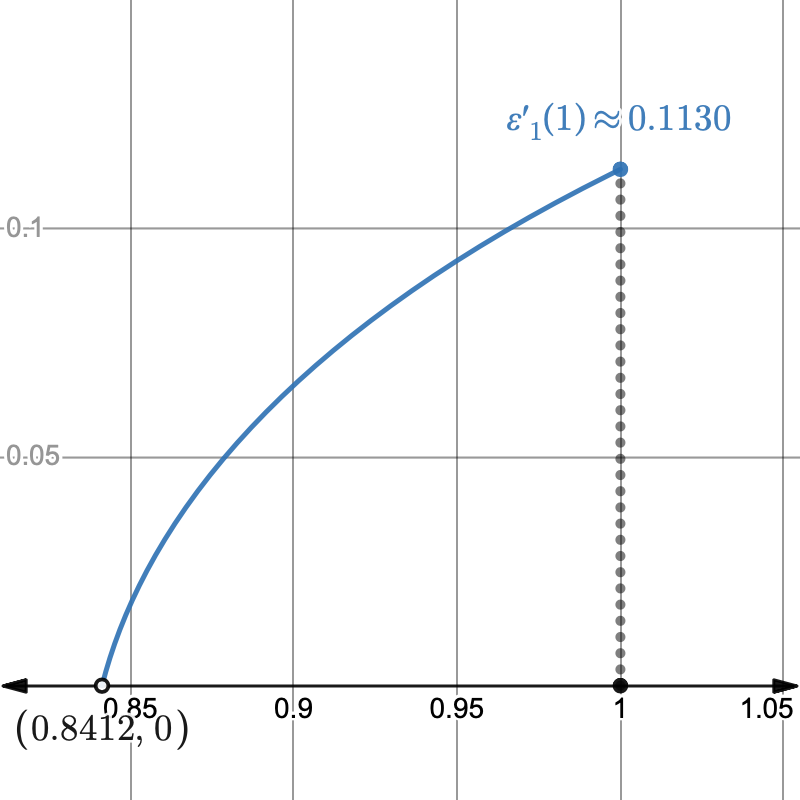}}}
    \caption{Expressions $\alpha_*(\rho)$, $\gamma_*(\rho)$, $\beta(\rho)$, $\epsilon'_1(\rho)$ in the density range $\frac{1}{2 - H(1/4)} < \rho \leq 1 + \Theta(\frac{1}{\log n})$. \\
    $\alpha_*(\rho)$ refers to the input-size-based speedup $\Omega(n^{\alpha})$, for $1/2 < \alpha < \alpha_*(\rho)$, over the runtime $O(2^n)$ of {\MeetInTheMiddle}. $\gamma_*(\rho)$ refers to the word-length-based speedup $\Omega(\word^{\gamma})$, for $0 < \gamma < \gamma_*(\rho)$, over the runtime $\tilde{O}(2^n \cdot n^{-1/2})$ of {\BitPacking}. Further, $\beta(\rho)$ and $\epsilon'_1(\rho)$ are the optimized parameters to achieve these speedups.}
    \label{fig:main}
\end{figure}
\clearpage}

\begin{proof}[Proof of Runtime for \blackref{alg:packed-rep-ov}]
Line~\ref{alg:packed-rep-ov:0} takes time $\poly(\word) = \poly(n)$ in the either case given \Cref{lemma:select-C,lemma:select-D}, which is dominated by runtime of the rest of our algorithm:

\vspace{.1in}
\noindent
\blackref{pack-rep-case1}{\bf .}
For the choices of parameters, the bottleneck occurs in Line~\ref{alg:packed-rep-ov:1} or Lines~\ref{alg:packed-rep-ov:4} to \ref{alg:packed-rep-ov:9}:
\ignore{For the choices of $e_1$, $e_2$, $|A|$, $|B|$, $|C|$ in this case:}
\begin{flushleft}
\begin{itemize}
    \item Line~\ref{alg:packed-rep-ov:1} takes time $\tO(2^{n/2} \cdot \word^{-1/2} \cdot \word^{-(1 - H(\frac{1 - \epsilon'_1}{2})) \cdot \rho / 4})$, by \Cref{lemma:input2-2} and the choice of $|C|$.
    
    \item Line~\ref{alg:packed-rep-ov:2} is skipped in this case.
    
    \item Line~\ref{alg:packed-rep-ov:3} takes time $O(2^{|A|} + 2^{|B|} + 2^{|C|} + 2^{|D|}) = \tO(2^{n/2} \cdot \word^{-1 / 2} \cdot \word^{-(2 - \rho + \beta) / 4})$, by \Cref{lem:sse} and the choices of $|A|$, $|B|$, $|C|$, $|D|$.
    This runtime is dominated by that of Line~\ref{alg:packed-rep-ov:1} as $(1 - H(\frac{1 - \epsilon_1'}{2}))\rho < (2 - \rho + \beta)$ for $\rho \leq 1$.
    
    \item Lines~\ref{alg:packed-rep-ov:4} to \ref{alg:packed-rep-ov:9} take time $\tO(k'(n,\, \word) \cdot \word^{1 - \rho + \beta/2}) = \tO(2^{n/2} \cdot \word^{-1/2} \cdot \word^{-((4 - 2H(1 / 4) - \epsilon'_1) \cdot \rho - 2 - \beta) / 4})$.
    
    Lines~\ref{alg:packed-rep-ov:4} iterates Lines~\ref{alg:packed-rep-ov:5} to \ref{alg:packed-rep-ov:7} a total of $|W(D) \bmod \bp| = \tO(\word^{1 - \rho + \beta/2})$ times, by \Cref{lemma:select-D}. \\
    Here each iteration takes time $\tO(\sum_{\boldr} (|\bR_{A,\, t',\, \boldr}| + |\bR_{B,\, t',\, \boldr}|)) = \tO(k'(n,\, \word))$, by \Cref{lemma:pack-main} and the classic {\em merge sort} algorithm (of depth $\log s'(\word) = O(\log \word) = O(\log n)$).
    
    Line~\ref{alg:packed-rep-ov:8} iterates Line~\ref{alg:packed-rep-ov:9} a total of $2^{|D|} = \word^{2 - \rho + \beta/2}$ times, by the choice of $|D|$. \\
    Here each iteration takes time $\tO(k'(n,\, \word) \cdot \word^{-1})$, by \Cref{lemma:MiM-main}.
\end{itemize}
\end{flushleft}
Thus we achieve a speedup of $\tOmega(\word^{1 / 2 + \gamma'_*})$, where the
\[
    \textstyle
    \gamma'_*
    ~:=~ \min\big\{\,
    \frac{1 - H((1 - \epsilon'_1) / 2)}{4} \cdot \rho,\quad
    \frac{4 - 2H(1/4) - \epsilon'_1}{4} \cdot \rho - \frac{1}{2} - \frac{\beta}{4} \,\big\}.
\]

\noindent
\blackref{pack-rep-case2}{\bf .}
For the choices of parameters, the bottleneck occurs in Line~\ref{alg:packed-rep-ov:1}, Line~\ref{alg:packed-rep-ov:2}, or Lines~\ref{alg:packed-rep-ov:4} to \ref{alg:packed-rep-ov:9}:
\ignore{$\lambda$, $e_1$, $e_2$, $|A|$, $|B|$, $|C|$, in this case:}
\begin{flushleft}
\begin{itemize}
    \item Line~\ref{alg:packed-rep-ov:1} takes time $\tO(2^{n/2} \cdot \word^{-1/2} \cdot \word^{-(1 - H(\frac{1-\epsilon_1}{2})) \cdot \beta/2})$, by \Cref{lemma:input2-2} and the choice of $|C|$.
    
    \item Line~\ref{alg:packed-rep-ov:2} takes time $\tO(2^{n/2} \cdot \word^{-1 / 2} \cdot \word^{-\lambda / 2})$, by \Cref{lemma:input1-2}.
    
    \item Line~\ref{alg:packed-rep-ov:3} takes time $O(2^{|A|} + 2^{|B|} + 2^{|C|} + 2^{|D|}) = \tO(2^{n/2} \cdot \word^{-1/2} \cdot \word^{-\beta/2})$, by \Cref{lem:sse} and the choices of $|A|$, $|B|$, $|C|$, $|D|$.
    This runtime is dominated by that of Line~\ref{alg:packed-rep-ov:1}.
    
    \item Lines~\ref{alg:packed-rep-ov:4} to \ref{alg:packed-rep-ov:9} take time $\tO(k''(n,\, \word)) = \tTheta(2^{n/2} \cdot \word^{-1/2} \cdot \word^{-((1 - H(\frac{1 + \epsilon_2}{4}) - \epsilon_1 / 2) \cdot \beta - \lambda)})$.
    
    Lines~\ref{alg:packed-rep-ov:4} iterates Lines~\ref{alg:packed-rep-ov:5} to \ref{alg:packed-rep-ov:7} a total of $|W(D) \bmod \bp| \leq |D| + 1 = O(\log\word) = O(\log n)$ times, \\
    because all integers in $\bD$ are congruent modulo $\bp$ by construction (Line~\ref{alg:packed-rep-ov:0}). \\
    Here each iteration takes time $\tO(\sum_{\boldr} (|\bR_{A,\, t',\, \boldr}| + |\bR_{B,\, t',\, \boldr}|)) = \tO(k''(n,\, \word))$, by \Cref{lemma:pack-main} and the classic {\em merge sort} algorithm (of depth $\log s''(\word) = O(\log \word) = O(\log n)$).
    
    Line~\ref{alg:packed-rep-ov:8} iterates Line~\ref{alg:packed-rep-ov:9} a total of $2^{|D|} = \word$ times, by the choice of $|D|$. \\
    Here each iteration takes time $\tO(k''(n,\, \word) \cdot \word^{-1})$, by \Cref{lemma:MiM-main}.
\end{itemize}
\end{flushleft}
Thus we achieve a speedup of $\tOmega(\word^{1 / 2 + \gamma''_*})$, where the
\begin{align*}
    \gamma''_* ~=~ \gamma''_*(\rho) & \textstyle ~:=~ \min\big\{\,
    \frac{\lambda}{2},\quad
    \frac{1 - H((1 - \epsilon_1) / 2)}{2} \cdot \beta,\quad
    (1 - H(\frac{1 + \epsilon_2}{4}) - \frac{\epsilon_1}{2}) \cdot \beta - \lambda \,\big\} \\
    & \textstyle ~\phantom{:}=~ \frac{\lambda}{2}
    ~=~ (1 - 10^{-5}) \cdot \frac{1 - \epsilon_1}{4} \cdot (1 - H(\frac{1 - \epsilon_2}{2})) \cdot \beta
    ~\geq~ 9.0324 \times 10^{-3} \cdot \beta.
\end{align*}
Here, we know that $\frac{\lambda}{2}$ is the minimum term for our chosen values of $\epsilon_1$, $\epsilon_2$ and $\lambda$ by \Cref{eq:lambda}.

\vspace{.1in}
\noindent
{\bf Worst Case.}
Overall, we achieve a speedup of $\Omega(\word^{1/2 + \gamma})$ for any constant $\gamma \in (0,\, \gamma_*)$. In particular, the constants $\gamma_* = \gamma_*(\rho) := \min \{ \gamma_*',\, \gamma_*''\}$, $\beta = \beta(\rho)$, $\epsilon'_1 = \epsilon'_1(\rho)$ are the solutions to the following chain of equations:
\begin{align}
\label{eq:main}
    \gamma_*(\rho)
    ~=~ 9.0324 \times 10^{-3} \cdot \beta(\rho)
    ~=~ \textstyle \frac{1 - H((1 - \epsilon'_1(\rho)) / 2)}{4} \cdot \rho
    ~=~ \textstyle \frac{4 - 2H(1/4) - \epsilon'_1(\rho)}{4} \cdot \rho - \frac{1}{2} - \frac{\beta(\rho)}{4}.
\end{align}
Alternatively, we can write the speedup as $\Omega(n^{\alpha})$ for $\alpha = \frac{1 + 2\gamma}{2\rho} \in (0,\, \alpha_*(\rho))$ and $\alpha_*(\rho) := \frac{1 + 2\gamma_*(\rho)}{2\rho}$.

In the density range $(2 - H(1/4))^{-1} < \rho \leq 1 + \Theta(\frac{1}{\log n})$, all expressions $\alpha_*(\rho)$, $\gamma_*(\rho)$, $\beta(\rho)$, $\epsilon'_1(\rho)$ are demonstrated in \Cref{fig:main}. In particular, the boundary $(2 - H(1/4))^{-1} \approx 0.8412$ is the solution to the equation $\alpha_*(\rho) = 1 / (2\rho)$. The minimum speedup is a factor of $\tOmega(n^{\alpha_*(1)}) \geq \Omega(n^{0.5023})$ for linear word length $\word = \Theta(n)$, and the maximum speedup is a factor of $\tOmega(n^{(2 - H(1/4)) / 2}) \geq \Omega(n^{0.5943})$ for slightly superlinear word length $\word = \Theta(n^{2 - H(1/4)}) \approx \Theta(n^{1.1887})$.

In the other density range $0 < \rho \leq (2 - H(1/4))^{-1} \approx 0.8412$, as mentioned, we can gain a better speedup $\tilde{\Omega}(n^{1/(2\rho)}) \geq \tOmega(n^{(2 - H(1/4)) / 2}) \geq \Omega(n^{0.5943})$ by naively simulating \blackref{alg:BitPacking}.
\end{proof}

\subsection{Auxiliary Lemmas}

\begin{lemma}[Finding a Subset $C$ with Large $W(C)$]
\label{lemma:select-C}
\begin{flushleft}
Given as input a modulus $p$ and a size-$n$ multiset $Y$ such that $|Y \bmod p| > \word^{\rho} / \log \word$, there always exists a size-$(\frac{1}{2}\log(\word^{\rho} / \log \word))$ subset $C \subseteq Y$ such that $|W(C) \bmod p| = 2^{|C|}$, and such a set $C$ can be found in time $\poly(\word) = \poly(n)$.
\end{flushleft}
\end{lemma}

\begin{proof}
We create the subset $C \subseteq Y$ using a simple greedy algorithm: starting with $C_{0} = \emptyset$, for $i \in [\frac{1}{2}\log(\word^{\rho} / \log \word)]$, select any $y_i \in (Y \setminus C_{i - 1})$ such that $y_i \not\equiv_{p} (c' - c'')$ for any two $c',\, c'' \in W(C_{i - 1})$, and then augment $C_{i} := C_{i - 1} \bigcup \{y_{i}\}$.
The existence of such an element $y_{i}$ at each stage is ensured, since any intermediate subset $C_{i - 1}$ has at most $|W(C_{i - 1})| \leq 2^{i - 1}\ignore{ \leq \sqrt{n / \log n}}$ subset sums, and thus $|\{(c' - c'') \bmod p \mid c',\, c'' \in C_{i - 1}\}| \leq \word^{\rho} / \log \word < |Y \bmod p|.$
\end{proof}

\begin{lemma}[Finding a Subset $D$ with Small $W(D)$]
\label{lemma:select-D}
\begin{flushleft}
Given as input a modulus $p = \Theta(\word^{1+\beta/2})$ and a size-$\Theta(n)$ multiset $Y$, there always exists a size-$((2 - \rho + \beta/2)\log \word)$ subset $D \subseteq Y$ such that $|W(D) \bmod p| = \tO(\word^{1 - \rho + \beta/2})$, and such a subset $D$ can be found in time $\poly(\word) = \poly(n)$.
\end{flushleft}
\end{lemma}

\begin{proof}
Given a small enough $q = \Theta(n/\log \word) = \Theta(\word^{\rho} / \log \word)$, partition $[p] = \bigcup_{j \in [q]} P_{j}$ into arithmetic progressions $P_j := \{r \in [p] \mid r \equiv_{q} j\}$ for $j \in [q]$, each of which contains $\lceil p / q \rceil$ or $\lfloor p / q \rfloor$ many residues.
By the pigeonhole principle, at least one of these progressions $P_{j^{*}}$ involves at least $|Y_{j^{*}}| \geq |Y| / q > (2 - \rho + \beta/2)\log \word$ many integers $Y_{j^{*}} \subseteq Y$, where $Y_{j^{*}} := \{y \in Y \mid y \in P_{j^{*}} \pmod{p}\}$.
We can find such a progression $P_{j^{*}}$ and the corresponding subset $Y_{j^{*}}$ in time $\poly(\word) = \poly(n)$.

We select $D = \{y_{i}\}_{i \in [|D|]}$ to be any $((2 - \rho + \beta/2)\log \word)$-integer subset of $Y_{j^{*}}$. Each selected $y_{i}$ has the residue $(y_{i} \bmod p) = q \cdot k_{i} + j^{*}$, where $k_{i} := \lfloor (y_{i} \bmod p) / q \rfloor \leq p/q$. Thus this subset $D$ has
\begin{align*}
    |W(D) \bmod p|
    \leq |W(\{k_{i}\}_{i \in [|D|]})| \cdot |D|
    \leq \Sigma(\{k_{i}\}_{i \in [|D|]}) \cdot |D|
    \leq p / q \cdot |D|^2
    = O(\word^{1 - \rho + \beta/2} \cdot \log^{3}(\word))
\end{align*}
distinct subset sums modulo $p$, as desired.
\end{proof}

Below, \Cref{lemma:input1-2,lemma:input2-2} are refinements of \Cref{lemma:input1,lemma:input2}, respectively, which leverage bit packing tricks used in the \blackref{alg:BitPacking} algorithm (recall \Cref{fig:alg:BitPacking}).

\begin{lemma}[Speedup via Additive Structure]
\label{lemma:input1-2}
Let $(X,\, t)$ be a Subset Sum instance. Given a subset $Y \subseteq X$ of size $|Y| \leq \frac{n - \log \word}{2}$ such that $|W(Y)| \leq 2^{|Y|} \cdot \word^{-\epsilon}$ for some constant $0 < \epsilon < |Y| / \log \word$\ignore{$\epsilon > 0$ satisfying $2^{|Y|} > \word^\epsilon$}, the instance $(X,\, t)$ can be solved in time $\tO(2^{n/2} \cdot \word^{-(1+\epsilon) / 2})$ with arbitrarily high constant probability and with no false positives.
\end{lemma}

\begin{proof}
Fix any partition $X = A \cup B \cup D$ such that $|D| = \log\word$, $|A| = \frac{n - (1-\epsilon)\log\word}{2}$ with $A \supseteq Y$, and $|B| = \frac{n - (1+\epsilon)\log\word}{2}$.
We have $|W(A)| \leq 2^{|A \setminus Y|} \cdot |W(Y)| \leq 2^{n / 2} \cdot \word^{-(1 + \epsilon) / 2}$ and $|W(B)| \leq 2^{n / 2} \cdot \word^{-(1 + \epsilon) / 2}$. Using \Cref{lem:sse,lemma:sse-poly,thm:bit-packing}, it takes expected time $O(2^{n/2} \cdot \word^{-(1+\epsilon) / 2} \cdot \log \word)$ to create the sorted lists $L_{A}$, $L_{B}$ and to run \blackref{alg:BitPacking}.

Cutting off the algorithm if it takes longer than a large constant times the expected runtime yields an one-sided error Monte Carlo algorithm by Markov's inequality.
\end{proof}

\begin{lemma}[Speedup via Unbalanced Solutions]
\label{lemma:input2-2} Let $(X,\, t)$ be a Subset Sum instance that has a solution. Given $Y \subseteq X$ of size $|Y| = c \log \word$ such that some solution $S \subseteq X$ satisfies $|S \cap Y| \notin (1 \pm \epsilon)\frac{|Y|}{2}$ for some constants $c > 0, \epsilon \in (0,\, 1)$, the solution $S$ can be found  in time $\tO(2^{n/2} \cdot \word^{-(1+\delta) / 2})$, where $\delta := (1 - H(\frac{1 - \epsilon}{2})) \cdot c$, with arbitrarily high constant probability.
\end{lemma}

\begin{proof}
The proof is a straightforward combination of the proofs of \Cref{lemma:input2,lemma:input1-2}.
\end{proof}

\begin{figure}[t]
    \centering
    \begin{mdframed}
    Subroutine $\term[\SampleList]{alg:sample-list}(\{R_{A,\, t',\, r}\}_{r}, h_m)$

    \begin{flushleft}
    {\bf Input:}
    A sorted list $R_{A,\, t'}$ containing elements in the sum-collection couple format $(a',\, \calQ_{a'})$, and a hash function $h_m$.
    
    
    \vspace{.05in}
    {\bf Output:}
    A hashed and packed list $H_{A,\, t'}$.
    
    \begin{enumerate}[label =
    7(\alph*)]
        \item\label{alg:sample-list:1}
        Initialize indices $i,\, j := 0$.
        While $i < |R_{A,\, t'}|$, set three words of $H_{A,\, t'}$ as follows:
        
        \item\label{alg:sample-list:2}
        \qquad $H_{A,\, t'}[3j] := a'_{i}$ stores the full integer $a'_{i}$, the smallest not-yet-packed sum.
        
        \item\label{alg:sample-list:3}
        \qquad $H_{A,\, t'}[3j+1] := ((h_{m}(a'_{i}),\, \calQ_{a'_{i}}),\, (h_{m}(a'_{i + 1}),\, \calQ_{a'_{i + 1}}),\, \dots)$ packs as many \\
        \qquad hash-collection couples as will fit into a single $\word$-bit word. \\
        \qquad Update $i$ to the index of the next not-yet-packed sum in $R_{A,\,t'}$.
        
        \item\label{alg:sample-list:4}
        \qquad $H_{A,\, t'}[3j+2] := a'_{i-1}$ stores the full integer $a'_{i-1}$, the largest already-packed sum.

        \item\label{alg:sample-list:5}
        \qquad Increment $j \gets j + 1$.
    \end{enumerate}
    \end{flushleft}
    \end{mdframed}
    \caption{The {\SampleList} subroutine.
    \label{fig:alg:sample-list}}
\end{figure}

\begin{lemma}[Hashing and Packing; Line~\ref{alg:packed-rep-ov:7}]
\label{lemma:pack-main}
In the context of \blackref{alg:packed-rep-ov}: \\
Line~\ref{alg:packed-rep-ov:7} uses the subroutine \blackref{alg:sample-list} to create lists $\bH_{\bA,\, t'}$ and $\bH_{\bB,\, t'}$ in time $\tO(k(n,\, \word))$ and to store these lists in $|\bH_{\bA,\, t'}| + |\bH_{\bB,\, t'}| = \tO(k(n,\, \word) \cdot \word^{-1})$ many words.
\end{lemma}

\begin{proof}
Line~\ref{alg:packed-rep-ov:7} creates $\bH_{\bA,\, t'}$, $\bH_{B,\, t'}$ using the subroutine \blackref{alg:sample-list} (\Cref{fig:alg:sample-list}).
The claimed runtime $\tO(k(n,\, \word))$ follows from the fact that the input list $\bR_{\bA,\, t'}$ contains $|\bR_{\bA,\, t'}| \leq \sum_{a'} |\bcalQ_{a'}| \leq k(n,\, \word)$ sum-collection couples by construction (\Cref{lemma:sse-quarter} and Line~\ref{alg:packed-rep-ov:5}), and \blackref{alg:sample-list} goes through it once, spending time $O(\log \word)$ on each couple to compute the hash $\bh_{m}(a')$.

To see that the output list $\bH_{\bA,\, t'}$ is stored in $\tO(k(n,\, \word) \cdot \word^{-1})$ many $\word$-bit words, we first observe that
a single hash-collection couple $(\bh_{m}(a'),\, \bcalQ_{a'})$ takes at most
\[
    m + |\bcalQ_{a'}| \cdot |\bC|
    ~\leq~ m + |\bcalQ^{+e_3}(\bC)| \cdot |\bC|
    ~\leq~ m + 2^{|\bC|} \cdot |\bC|
    ~=~ \tO(\word^{1/2})
    ~=~ o(\word)
\]
many bits, for $m = 3\log \word$ and $|\bC| =  \max\{\frac{1}{2}\log(\frac{\word^{\rho}}{\log \word}),\, \beta\log(\frac{\word}{\beta\log \word})\} \leq \frac{1}{2}\log \word$ (Line~\ref{alg:packed-rep-ov:0} and \Cref{fig:main:beta}).
That is, a single word $\bH_{\bA,\, t'}[3j+1]$ created in Line~\ref{alg:sample-list:3} has at most $o(\word)$ many unused bits, which is negligible compared to the word length. 
Moreover, all of the $|\bR_{\bA,\, t'}| \leq \sum_{a'} |\bcalQ_{a'}| \leq k(\word)$ many hash-collection couples $(\bh_{m}(a'),\, \bcalQ_{a'})$ take a total of at most
\[
    \textstyle
    |\bR_{\bA,\, t'}| \cdot m + \sum_{a'} |\bcalQ_{a'}| \cdot |\bC|
    ~\leq~ k(n,\, \word) \cdot 3\log \word + k(n,\, \word) \cdot \frac{1}{2}\log \word
    ~=~ \tO(k(n,\, \word))
\]
many bits. Accordingly, a total of $\tO(k(n,\, \word)) \cdot \word^{-1} = \tO(k(n,\, \word) \cdot \word^{-1})$ many words $\bH_{\bA,\, t'}[3j+1]$ are created throughout all the executions of Line~\ref{alg:sample-list:3}. The entire output list $\bH_{\bA,\, t'}$ has three times as many words because of the additional elements storing the highest and the lowest values hashed into each word $\bH_{\bA,\, t'}[3j+1]$.
\end{proof}

\begin{figure}[t]
    \centering
    \begin{mdframed}
    Subroutine $\term[\SearchList]{alg:search-list}(t'',\, R_{A,\, t'},\, R_{B,\, t'},\, H_{A,\, t'},\, H_{B,\, t'})$

    \begin{flushleft}
    {\bf Input:}
    A shifted target $t''$ and sorted lists $R_{A,\, t'}$, $R_{B,\, t'}$, $H_{A,\, t'}$, $H_{B,\, t'}$ for $t' = (t'' \bmod p)$.
    
    \begin{enumerate}[label =
    9(\alph*)]
    \item\label{alg:sample-searching:1}
    Initialize indices $i := 0$ and $j := |H_{B,\,t'}|-1$. While $i < |H_{A,\,t'}|$ and $j \geq 0$:
    
    \item\label{alg:sample-searching:2}
    \qquad If the indexed words $H_{A,\,t'}[3i+1]$, $H_{B,\,t'}[3j+1]$ contain hash-collection couples \\ 
    \qquad $(\bh_{m}(a'),\, \calQ_{a'})$, $(\bh_{m}(b'),\, \calQ_{b'})$ with $\bh_{m}(a') + \bh_{m}(b') \in \bh_{m}(t'') - \{0,\, 1\} \pmod{2^{m}}$ \\
    \qquad and $Q_{a'} \cap Q_{b'} = \emptyset$ for some $Q_{a'} \in \calQ_{a'}$, $Q_{b'} \in \calQ_{b'}$, use \Cref{lemma:sse-quarter} to search the \\
    \qquad 
    corresponding sublists of $R_{A,\, t'}$, $R_{B,\, t'}$ for a solution as described in the proof of \\
    \qquad  \Cref{lemma:MiM-main}. Halt and return ``yes'' if a solution is found.
    
    \item\label{alg:sample-searching:3}
    \qquad If $H_{A,\,t'}[3i] + H_{B,\,t'}[3j+2] < t''$, set $i \gets i + 1$. Otherwise, set  $j \gets j - 1$.
    \end{enumerate}
    \end{flushleft}
    \end{mdframed}
    \caption{The {\SearchList} subroutine.
    \label{fig:alg:search-list}}
\end{figure}

\begin{lemma}[Searching for Solutions; Line \ref{alg:packed-rep-ov:9}]
\label{lemma:MiM-main}
In the context of \blackref{alg:packed-rep-ov}: \\
{\bf (i)}~Line~\ref{alg:packed-rep-ov:9} finds a solution via \blackref{alg:search-list}, i.e., a pair of sum-subset couples $(a',\, Q_1)$, $(b',\, Q_2)$ with $a' + b' = t''$ and $Q_1 \cap Q_2 = \emptyset$, with arbitrarily high constant probability, if the lists $\bR_{\bA,\, (t'' \bmod \bp)}$ and $\bR_{\bB,\, (t'' \bmod \bp)}$ contain such a solution.
{\bf (ii)}~\blackref{alg:search-list} has 
runtime $\tO(k(n,\, \word) \cdot \word^{-1})$.
\end{lemma}

\begin{proof}
\blackref{alg:search-list} adapts the \blackref{alg:BitPacking} algorithm (Lines \ref{alg:BitPacking:4} to \ref{alg:BitPacking:6}) for the lists $\bH_{\bA,\, t'}$, $\bH_{\bB,\, t'}$ to address two additional issues.
First, not only a hash collision but also a pair of overlapping near-quartersets in $\bH_{\bA,\, t'}$, $\bH_{\bB,\, t'}$ may incur a false positive. Also, each word in $\bH_{\bA,\, t'}$, $\bH_{\bB,\, t'}$ packs an uncertain amount of hash-collision couples (while \blackref{alg:BitPacking} processes exactly $(\word / m)$ hashes per iteration).

Line~\ref{alg:sample-searching:2} settles the first issue by modifying the ``If'' test in line \ref{alg:BitPacking:4} of \blackref{alg:BitPacking}. This new test is implemented by the $(2\word)$-bit Boolean function \term[\texttt{Hash-OV}]{function:packed-OV}: it takes as input two packed words
$u = ((\bh_m(a'_1),\, \bcalQ_{a'_1}),\, (\bh_m(a'_2),\, \bcalQ_{a'_2}),\, \dots)$ and $v = ((\bh_m(b'_1),\, \bcalQ_{b'_1}),\, (\bh_m(b'_2),\, \bcalQ_{b'_2}),\, \dots)$, where each $\bh_m(a')$, $\bh_m(b')$ is a hashed value and each $\bcalQ_{a'},\, \bcalQ_{b'} \subseteq \bcalQ^{+e_3}(\bC)$ is a collection of near-quartersets.
$\blackref{function:packed-OV}(u,\, v)$ returns $1$ if and only if two conditions hold:
first, there is a pair of hash-collection couples $(\bh_{m}(a'), \calQ_{a'})$, $(\bh_{m}(b'), \calQ_{b'})$ with $\bh_{m}(a') + \bh_{m}(b') \in \bh_{m}(t'') - \{0,\, 1\} \pmod{2^{m}}$;
second, there are two disjoint near-quartersets $Q_{a'} \cap Q_{b'} = \emptyset$, for $Q_{a'} \in \calQ_{a'}$, $Q_{b'} \in \calQ_{b'}$, in the packed collections indexed by $\bh_m(a')$, $\bh_m(b')$.
\blackref{function:packed-OV} is an $\mathsf{AC}^0$ operation on two words and takes constant time to evaluate in the circuit RAM model, since all hash-collection couples can be checked in parallel.

Now on comparing any two words $\bH_{\bA,\, t'}[3i+1]$, $\bH_{\bB,\, t'}[3j+1]$, by the second condition above, overlapping near-quartersets never incur false positives. Hence like \blackref{alg:BitPacking}, the new ``If'' test in Line~\ref{alg:sample-searching:2} is passed by every correct solution and (some of) the hash collisions.
We further verify such a potential solution by returning back to the (unhashed) lists $\bR_{\bA,\, t'}$, $\bR_{\bB,\, t'}$ and checking all the sum-collection couples that are packed into $\bH_{A,\,t'}[3i+1]$, $\bH_{\bB,\,t'}[3j+1]$, namely the two sorted sublists
$\{(a',\, \bcalQ_{a'}) \mid \bH_{\bA,\, t'}[3i] \leq a' \leq \bH_{\bA,\, t'}[3i+2]\}$ and
$\{(b',\, \bcalQ_{b'}) \mid \bH_{\bB,\, t'}[3j] \leq b' \leq \bH_{\bB,\, t'}[3j+2]\}$.
Either sublist is stored in at most $(\word / m)$ words by construction (see \blackref{alg:sample-list},  Line~\ref{alg:sample-list:3}).
Hence by \Cref{lemma:sse-quarter}, this verification process takes time $O(\word / m) = \poly(n)$.

Line~\ref{alg:sample-searching:3} settles the second issue by replacing the test in line \ref{alg:BitPacking:6} of \blackref{alg:BitPacking} with a new test:
If $\bH_{\bA,\, t'}[3i] + \bH_{\bB,\, t'}[3j+2] < t''$? By construction (\blackref{alg:sample-list}), $\bH_{\bA,\,t'}[3i]$ and $\bH_{\bB,\,t'}[3j+2]$ are the exact values of the smallest sum packed into $\bH_{\bA,\,t'}[3i+1]$ and the largest sum packed into $\bH_{\bB,\,t'}[3j+1]$, respectively. By the same argument as in the proof of correctness for \blackref{alg:BitPacking}, in a single scan of $\bH_{\bA,\, t'}$ and $\bH_{\bB,\, t'}$, we never miss a pair of words that packs a correct solution.

\vspace{.1in}
\noindent
{\bf (ii).}~\blackref{alg:search-list} performs a single scan of $\bH_{\bA,\, t'}$ and $\bH_{\bB,\, t'}$, plus the verification of at most one correct solution versus the hash collisions. Using the same argument as in the proof of runtime for \blackref{alg:BitPacking} (Line~\ref{alg:BitPacking:5}), the expected verification-time $\poly(n) + (|\bH_{\bA,\, t'}| + |\bH_{\bB,\, t'}|) \cdot o_{n}(1)$ is dominated by the scan-time $= O(|\bH_{\bA,\, t'}| + |\bH_{\bB,\, t'}|) = \tO(k(n,\, \word) \cdot \word^{-1})$, given \Cref{lemma:pack-main}.
\end{proof}

\begin{observation}[Adapting \blackref{alg:packed-rep-ov} to Word RAM]
\label{obs:packed-rep-ov-word-RAM}
As in \Cref{sec:memo-ov}, to adapt \blackref{alg:packed-rep-ov} to the word RAM model, we run the algorithm as if the word length were $\word' := 0.1n$ and memoize \blackref{function:packed-OV} to speed up the evaluation of this function. The resulting runtime is $O(2^{n/2} \cdot n^{-0.5023})$ for $(2 - H(1/4))^{-1} < \rho < 1 + \Theta(1 / \log n)$.
There are three additional modifications:
\begin{enumerate}
    
    \item Line~\ref{alg:sample-list:1} sets $\lceil \word' / \word \rceil + 2$ words (rather than three words) in a single iteration, 
    i.e., the single word set in Line~\ref{alg:sample-list:3} is replaced with $\lceil \word' / \word \rceil = \Theta(1)$ words (since a single word with $\word$ bits may be insufficient).
    \item Before the execution of Line~\ref{alg:packed-rep-ov:8}, create a table $\texttt{Hash-OV}'$ that memoizes all input-output results of the Boolean function \blackref{function:packed-OV}, in time $(2^{\word'})^{2} \cdot \poly(\word') = O(2^{0.21n})$\ignore{$(3^{\word'})^{2} \cdot \poly(\word') = 2^{(0.2\log3)n} \cdot \poly(n) = O(2^{0.32n})$}. We also replicate the table $\texttt{OV}'$ described in \Cref{obs:memo-ov-word-RAM}. Then either table can be accessed using a $2\lceil \word' / \word \rceil = \Theta(1)$-word index in constant time.
    \item Line~\ref{alg:sample-searching:2} replaces the functions \blackref{function:packed-OV}
    (used in \Cref{lemma:MiM-main}) and \blackref{function:OV} (used in \Cref{lemma:sse-quarter}) with constant-time lookup into the above memoized tables $\texttt{Hash-OV}'$ and $\texttt{OV}'$, respectively.
\end{enumerate}
Compared with running \blackref{alg:packed-rep-ov} itself for $\word' = 0.1n$, the only significant difference of this word RAM variant is that the number of packed words in $\bH_{A, t'}$ and $\bH_{B, t'}$ increases by a $\Theta(1)$ factor, so we can easily check the correctness and the (asymptotically) same runtime.\ignore{and correctness and runtime are otherwise as in the $\word = \Theta(n)$ case.}
\end{observation}

\section{Extensions and Future Work}
\label{sec:extensions}

Our results open up several natural directions for future investigation:

\begin{itemize}
    \item {\bf Derandomization:}  All of the $2^{n/2}/\poly(n)$-time algorithms we have given for Subset Sum use randomness.
    Can our results be extended to achieve deterministic algorithms with worst-case running time $2^{n/2}/\poly(n)$?
    
    \item {\bf Counting:} It is straightforward to modify the  \blackref{alg:MeetInTheMiddle} algorithm to output a count of the number of Subset Sum solutions in time $O(2^{n/2})$ (essentially, by keeping track of the multiplicity with which each value occurs in each list $L_A,L_B$). Can our techniques be extended to give counting algorithms for Subset Sum that run in time $2^{n/2}/\poly(n)$? In time $2^{n/2}/n^{0.501}$?
    
    \item {\bf Faster runtimes:} Finally, an obvious goal is to quantitatively strengthen our results by developing faster algorithms for worst-case Subset Sum. It would be particularly interesting to achieve running times of the form $2^{n/2}/f(n)$ for some $f(n)=n^{\omega(1)}$.
\end{itemize}




\section*{Acknowledgements}
We would thank Martin Dietzfelbinger for pointing out the work \cite{dietzfelbinger1997reliable}.

\bibliographystyle{alpha}
\bibliography{main}
    
\appendix
    
\end{document}